%% file: fsttcs_main.tex
\documentclass[a4paper,UKenglish,cleveref, autoref, thm-restate]{lipics-v2021}

\hideLIPIcs  

\usepackage{xspace}
\usepackage{apxproof}
\usepackage{caption}

\usepackage[linesnumbered,ruled,vlined]{algorithm2e}
\usepackage{todonotes}
\usepackage{macros}
\usepackage{tikz}
\usetikzlibrary{arrows.meta, positioning}
\usetikzlibrary{arrows,automata,positioning, shapes.geometric}
\usepackage{multirow}
\usepackage{makecell}

\SetKwInput{KwInput}{Input}
\SetKwInput{KwOutput}{Output}

\title{Scalable Learning of One-Counter Automata via State-Merging Algorithms}

 \author{Shibashis Guha}{School of Technology and Computer Science, Tata Institute of Fundamental Research, India \and \url{https://www.tifr.res.in/shibashis.guha/}}{shibashis@tifr.res.in}{https://orcid.org/0000-0002-9814-6651}{} \author{Anirban Majumdar}{Independent Researcher \and \url{https://anirban11.github.io/}}{majumdaranirban963@gmail.com}{https://orcid.org/0000-0003-4793-1892}{}
 \author{Prince Mathew}{School of Mathematics and Computer Science, Indian Institute of Technology Goa, India \and \url{https://princemathew07.github.io/} }{prince@iitgoa.ac.in}{https://orcid.org/0000-0001-6410-1474}{}
  \author{A.V. Sreejith}{School of Mathematics and Computer Science, Indian Institute of Technology Goa, India \and \url{https://www.iitgoa.ac.in/~sreejithav/} }{sreejithav@iitgoa.ac.in}{}{}

 \authorrunning{S. Guha, A. Majumdar, P. Mathew, A.V. Sreejith} 



\acknowledgements{Prince Mathew and A.V. Sreejith would like to thank CEFIPRA for the project IFC/6602-1/2022/77. 
A.V. Sreejith would also like to thank the DST Matrics grant for the project MTR/2022/000788.}

\nolinenumbers 

\begin{document}

\maketitle

\begin{abstract}
We propose One-counter Positive Negative Inference (\OPNI), a passive learning algorithm for deterministic real-time one-counter automata (\DROCA). Inspired by the \RPNI algorithm for regular languages, \OPNI constructs a \DROCA consistent with any given valid sample set.

We further present a method for combining \OPNI with active learning of \DROCA, and provide an implementation of the approach. Our experimental results demonstrate that this approach scales more effectively than existing state-of-the-art algorithms. We also evaluate the performance of the proposed approach for learning visibly one-counter automata. 

\end{abstract}

\input{introduction}

\input{prelims}

\input{opni}

\input{active-passive}

\input{experiments}

\input{conclusion}

\bibliography{biblio}
\newpage
\end{document}

%% file: introduction.tex
\section{Introduction}
\label{sec:intro}
{\bf Automata learning and verification.}
Automata learning constitutes a correct-by-construction synthesis technique aimed at inferring formal models—such as finite-state machines or Mealy machines—from observed system behaviour. It ensures that the inferred model remains consistent with the observed data. Closely linked to formal verification, automata learning plays a key role in model inference, facilitating formal analysis and validation of system behaviour. It may also be viewed as an inductive synthesis framework tailored for the construction of finite-state programs.

\noindent{\bf Active and passive learning of automata.}
The $L^*$ algorithm~\cite{Angluin87} of Angluin is a foundational algorithm in active automata learning. It learns a minimal deterministic finite automaton (\dfa) that accepts an unknown regular language, using a query-based learning model. The learner interacts with a teacher (oracle) through two types of queries, called membership queries and equivalence queries.
The $L^*$ algorithm guarantees the learning of the minimal \dfa in a number of steps polynomial in the size of the minimal \dfa and the length of the longest counterexample.

Passive learning of \dfa, for example, Regular Positive and Negative Inference (\RPNI)~\cite{oncina1992}, on the other hand, involves inferring a \dfa from a fixed, finite dataset consisting of positive and negative examples of strings.
The output of the algorithm is a \dfa that accepts all positive examples and rejects all negative ones.
The constructed \dfa may depend on the input positive and negative examples. For arbitrary inputs, the \dfa that is consistent with the input data may not be minimal.

\noindent{\bf One-counter automata and visibly one-counter automata.}
Deterministic one-counter automata~\cite{VP75} are a subclass of deterministic pushdown automata that operate with a single integer counter, which can be incremented, decremented, or tested for zero during transitions. They serve as a simple model for programs with minimal memory—more powerful than finite automata, but less expressive than general deterministic pushdown automata.
Deterministic real-time one-counter automata (\droca)~\cite{droca} form a subclass of deterministic one-counter automata that do not have $\varepsilon$-transitions.

Visibly pushdown automata (\textsc{vpa})~\cite{AM04,AM09} are a restricted subclass of pushdown automata where the type of stack operation (push, pop, or no operation) is determined by the input symbol. This restriction allows \textsc{vpa} to retain much of the expressive power of pushdown automata while ensuring desirable closure properties and algorithmic properties similar to finite automata.
Unlike pushdown automata, deterministic visibly pushdown automata are equally expressive as nondeterministic ones.
Deterministic visibly one-counter automata (\voca) are both \droca and deterministic \textsc{vpa}.

In this paper, we study active and passive learning of \droca, and extend them for \voca. 

\noindent{\bf Contributions.}
We present two new methods for learning \droca, addressing both passive and active learning settings:
\begin{itemize}
\item Our first contribution (Section \ref{sec:opni}) is a passive learning algorithm, called \opni. Given a labelled set $\S$ of words—partitioned into accepting and rejecting examples—along with counter values for all prefixes, \opni synthesises a \droca\ that is consistent with both the counter information and the acceptance labels.
\item Our second contribution (Section \ref{sec:active-passive}) is an active learning procedure, \opniL.  In this setting, the learner interacts with a teacher who possesses a target \droca. The learner may issue membership, equivalence, and counter-value queries to infer the target automaton. \opniL adapts the \minOCA framework~\cite{MathewPS25} by replacing its SAT-based subroutine with our passive learner \opni, leading to a simpler and significantly more scalable method.
\end{itemize}

We implemented \opniL in Python and evaluated it on randomly generated \droca. Our experiments (Section \ref{sec:experiments}) show that \opniL\ outperforms the state of the art \minOCA algorithm in terms of scalability: while \minOCA fails on most \droca\ with greater than 12 states, \opniL successfully learns \droca\ of much larger size.
Additionally, we evaluated \opniL for the class of \voca. 
In this case, \opniL demonstrates remarkable scalability, learning most \voca\ with up to $60$ states and nearly half with $100$ states. In contrast, \minOCA fails to learn most \voca\ with greater than 20 states.

However, unlike \minOCA, \opniL\ need not necessarily learn the minimal \droca. 
    
\noindent{\bf Related works.}
In the passive learning framework, the learner receives a set of labelled examples and synthesises a model consistent with them. The classical \RPNI algorithm~\cite{oncina1992} learns a \dfa in polynomial time from a set of accepting and rejecting words.

Active learning of finite automata via membership and equivalence queries was pioneered by Angluin's $L^*$ algorithm~\cite{Angluin87}.  
In the context of learning one-counter automata, Fahmy and Roos~\cite{FahmyR95} established the decidability of learning \droca. Later, Neider and L\"oding \cite{christof}  proposed an algorithm for learning \voca, employing an additional partial-equivalence query. Building on Neider and L\"oding's techniques, Bruy\'ere et al. \cite{bps} proposed learning \droca\ by combining counter-value queries—which return the counter value after processing a word. Mathew et al. \cite{MathewPS25} proposed an alternate algorithm (\minOCA) to learn \droca.
It employs SAT solvers to compute a minimal separating \dfa. These algorithms have worst-case exponential runtime. Recent work~\cite{LearningInP} shows that active learning of \droca is possible in polynomial time.

\noindent {\bf Organisation.} The remainder of the paper is organised as follows. Section~\ref{sec:prelims} presents the necessary preliminaries. In Section~\ref{sec:opni}, we introduce our passive learning algorithm, \opni. Section~\ref{sec:active-passive} describes the active learning procedure, \opniL. Section~\ref{sec:experiments} reports on our experimental evaluation of \opniL and compares its performance with the state-of-the-art \minOCA algorithm. Finally, Section~\ref{sec:conclusion} concludes the paper.

%% file: prelims.tex
\section{Preliminaries}
\label{sec:prelims}
For {a} finite set $S$, {we denote by} $|S|$ the cardinality {of $S$}. Non-negative integers are denoted by $\N$, and 
$[i,j]$ denotes the interval $\{i, i+1, \ldots , j\} \subseteq \N$. 
The sign of a non-negative integer $d$ (denoted by $\sgn(d)$) is $0$ if $d=0$ and $1$ otherwise. 

An alphabet denoted by $\Sigma$ {is a finite set of letters} and $\Sigma^*$ (resp., $\Sigma^+$) represents the set of all words including (resp., excluding) the empty word $\varepsilon$ over the alphabet $\Sigma$. For a word $w=a_0a_1a_2\ldots a_n$ and non-negative integers $i < j$, we use $w[i \cdots j]$ for the factor {$a_i a_{i+1} \cdots a_j$} and $w[i]$ for the letter $a_i$. 
Given a set {$S \subseteq \Sigma^*$} of words, we write $\pref{S}$ to denote the set of prefixes of all words in $S$. 

We now define the \emph{length-lexicographic order} (denoted by $\llex$) on words. First, we fix an arbitrary total order $\lex$ on the letters in $\Sigma$. The order of words is inductively defined as follows:
\(
u \llex v \quad \text{if either} \quad |u| < |v|, \quad \text{or} \quad |u| = |v| \text{ and there exists } x, y, z \in \Sigma^* \text{ and } \sigma_1, \sigma_2 \in \Sigma \text{ such that } u = x \sigma_1 y,\ v = x\sigma_2 z \text{ and }  \sigma_1 \lex \sigma_2.
\)
Note that this ordering is total and well-founded. Note also that this order can naturally be extended to pairs of words as follows: \( (u_1,v_1) \llex (u_2, v_2) \) if either \( u_1 \llex u_2 \), or \( u_1 = u_2\) and \( v_1 \llex v_2 \).

    A \emph{deterministic finite automaton} (\dfa) is a tuple \( D = (Q, \Sigma, \delta, q_0, F) \), where
    \( Q \) is a finite set of states,
    \( \Sigma \) is a finite input alphabet,
    \( \delta \colon Q \times \Sigma \to Q \) is the transition function,
    \( q_0 \in Q \) is the initial state, and
    \( F \subseteq Q \) is the set of accepting (final) states.

Given a \dfa $D$, we sometimes write $q\xrightarrow{a}q'$ to denote $\delta(q,a) = q'$. The transition function \( \delta \) extends naturally to words in \( \Sigma^* \) in the usual way: for \( w = a_1 a_2 \cdots a_n \in \Sigma^* \), we define \( \delta(q, w) = \delta(\cdots\delta(\delta(q, a_1), a_2), \dots, a_n) \). We will write \( D(w) \) to denote \( \delta(q_0,w) \). A word \( w \in \Sigma^* \) is \emph{accepted} by \( D \) if \( D(w) \in F \). The language recognised by \( D \), denoted as \( \Lang(D) \), is the set of all accepted words. 
Two \dfas $D_1$ and $D_2$ are equivalent if $\Lang(D_1)=\Lang(D_2)$.

\begin{definition}[\droca]
\label{def:droca}
    A \emph{deterministic real-time one-counter automaton} (\droca) is a tuple $\Autom= (Q,\Sigma, q_0,\delta_0,\delta_1,F)$, where $Q$ is a finite nonempty set of states, $\Sigma$ is the input alphabet, $q_0\in Q$ is the initial state, $\delta_0: Q \times \Sigma \to Q \times \{0,+1\}$ and $\delta_1: Q \times \Sigma \to Q \times \{0,+1,-1\}$ are the transition functions, and $F\subseteq Q$ is the set of final states.
\end{definition}

Consider a \droca $\Autom$. A configuration of $\Autom$ is a pair $(q, n)\in Q \times \N$, where $q$ denotes the current state and $n$ is the counter value. The configuration $(q_0,0)$ is called the \emph{initial configuration} of $\Autom$. For an $e \in \{-1,0,+1\}$ and letter $a$, the \emph{transition} between the configurations $(p,n)$ and $(q,n+e)$ on the symbol $a$ is defined if $\delta_{\sgn(n)}({p},a) = ({q},e)$. We use $(p,n)\xrightarrow{a}(q,n+e)$ to denote this. 
The run of a word $w=a_1\dots a_n$ on $\Autom$, if it exists, is the sequence of configurations $(q_0,m_0) \xrightarrow{a_1} (q_1,m_1) \dots \xrightarrow{a_{n-1}} (q_n,m_n)$ where $m_0 = 0$. We will write $(q_0,m_0) \xrightarrow{w} (q_n,m_n)$ to denote such a sequence. Further, we say $m_n$ is the \emph{counter-effect} of $w$ (denoted by $\ce_{\Autom}(w)$). 
Note that the counter values always stay non-negative, implying a decrement is not permitted from a configuration with a zero counter value. 
We will say that $w$ is accepted by $\Autom$ if and only if $q_n \in F$. 
{A \droca $\Autom$ is \emph{complete} if, for {every} word $w$, there is exactly one run on $w$.}

The language of $\Autom$, denoted by $\Lang(\Autom)$, is the set of all words accepted by $\Autom$. Similar to the case of \dfa, we will write \( \Autom(w) \) to denote the state in $\Autom$ that is reached after reading $w$. 
The definition of language equivalence is also similar to that above. Note that there are no $\epsilon$-transitions in a \droca. 
We say two \drocas \Autom and \Butom are \emph{counter-synchronous} if \upshape$\ce_{\Autom}(w)=$\upshape$\ce_{\Butom}(w)$ for all words $w$.

A \voca is a \droca where the input alphabet $\Sigma$ is a union of three disjoint sets $ \Sigma_{call}, \Sigma_{ret},$ and  $\Sigma_{int}$. The \voca  increments (resp.~decrements) its counter on reading a symbol from $\Sigma_{call}$ (resp.~$\Sigma_{ret}$). The counter value is unchanged on reading a symbol from $\Sigma_{int}$. 

\begin{definition}[\voca]
\label{defdroca}
A visibly one-counter automaton (\voca) is a tuple $\Autom= (Q,\Sigma_{call} \cup \Sigma_{ret} \cup \Sigma_{int} , q_0,\delta_0,\delta_1,F)$, where 
$Q$ is a finite nonempty set of states, 
$\Sigma= \Sigma_{call} \cup \Sigma_{ret} \cup \Sigma_{int}$ is the input alphabet, 
$q_0\in Q$ is the initial state, 
$\delta_0: Q \times \Sigma \to Q \times \{0,+1, -1\}$, $\delta_1: Q \times \Sigma \to Q \times \{0,+1,-1\}$ are the transition functions, and 
$F\subseteq Q$ is the set of final states.
\end{definition}

The notions of counter-effect, runs and transitions for \voca remain the same as those of \droca. From the definition, \voca are deterministic. 
The counter-effect of a transition is solely based on $\Sigma$, making the starting state and counter value irrelevant.
For $\sigma\in\Sigma$:
\[\ce_{\Autom}(\sigma)= \begin{cases}
1, & \text{ if } \sigma\in \Sigma_{call}\\
-1, & \text{ if } \sigma\in \Sigma_{ret},\text{ and}\\
0, & \text{ if } \sigma\in \Sigma_{init}\\
\end{cases}\]
If $w=\sigma_1\sigma_2\ldots \sigma_n$ for some $n\in\N$ and $\sigma_1,\ldots,\sigma_n\in\Sigma$, then $\ce_{\Autom}(w)=\ce_{\Autom}(\sigma_1)+\ce_{\Autom}(\sigma_2)+\ldots +\ce_{\Autom}(\sigma_n)$. Since the counter value can never go below zero in an \oca, there will be words that do not have a valid run in a \voca.  These words will be considered as rejected. 

\subsection{\RPNI: a passive learning algorithm for DFA}
\label{subsec:rpni}
\subparagraph*{Learning framework.} 
A pair of sets of words $\S = \splus \cup \sminus$ over an alphabet $\Sigma$ is called a \emph{sample set}, and words in $\splus$ (resp., $\sminus$) are called \emph{positive} (resp., \emph{negative}) samples. 
A sample set $\S$ is called \emph{consistent} if it satisfies the condition: $\splus \cap \sminus = \emptyset$.
Typically, in a \emph{passive learning} framework for \dfas, a \emph{learner} is given an alphabet $\Sigma$, and a sample set $\S$ {over $\Sigma$}, that is consistent, and the objective is to construct an automaton $D$ such that $D$ accepts all the words in $\splus$ and rejects all the words in $\sminus$. 
In that case, abusing notation, we will say that $D$ is \emph{consistent} with the sample set $\S$.

\subparagraph*{\RPNI algorithm.}
Here we recall \RPNI, a state-of-the-art passive learning algorithm for \dfas, originally proposed in~\cite{oncina1992}. Intuitively, the algorithm starts with a prefix tree acceptor (PTA) constructed from the positive samples $\splus$, and then iteratively merges states of the PTA in a specific order (precisely, the length-lexicographic order), keeping the resulting \dfa consistent with the samples set $\S$.  
A pseudo-code of our version of \RPNI is given in~\cref{alg:rpni}. Below, we explain different steps of the algorithm.

\begin{algorithm}[t]
    \caption{\RPNI: a passive learning algorithm for \dfa}
    \label{alg:rpni}
    \KwInput{An alphabet $\Sigma$, and a \emph{consistent} sample set $\sample = \splus \cup \sminus$ over $\Sigma$.}
    \KwOutput{A \dfa \( D \) consistent with the sample.}
    \BlankLine
    \emph{Construct PTA}: Build a prefix tree acceptor \( \PT \) from \( \splus \).\\
    \emph{Pairing}: \( P = \{(u,v) ~|~ v \llex u \) and $u, v \in \pref{\splus}$\}. \\
    \emph{Ordering P}: Sort the elements of $P$ according to the increasing \llex\ order. \\
    Initialize $D \leftarrow \PT$.\\
    \While{P is not empty}{
        Pop $(u,v)$ the minimal element from $P$.\\
        \If{$D(u)$ and $D(v)$ can be merged} {
            $D \leftarrow$ Merge$(D;D(u),D(v))$.
        }
    }
    \Return \( D \).
    \end{algorithm}

We first construct a PTA (line 1) from the positive samples $\splus$, which can formally be defined as a \dfa \( \PT = (Q, \Sigma, \delta, q_\varepsilon, F) \), where \( Q = \{q_u \mid u \in \pref{\splus} \} \) represents the set of states corresponding to all prefixes of words in $\splus$, \( \delta \) is defined as follows: for all $u, u\sigma \in \pref{\splus}$ with $\sigma\in \Sigma$, we have $\delta(q_u, \sigma) = q_{u\sigma}$, the initial state $q_{\varepsilon}$ is the state corresponding to the empty word $\varepsilon$, and \( F = \{q_u \mid u \in \splus\}\) is the set of all states that correspond to words in $\splus$. We call a prefix $u$ the representative of the state $q_u$. Note that $\PT$ accepts exactly the words in $\splus$.

We then aim to merge states of the PTA in the $\llex$ ordering of their representatives, while reserving the consistency of the resulting \dfa with the negative samples. 
To this end, we consider all pairs of prefixes \( (u,v )\) in $\pref{\splus}$ such that $u$ is greater than $v$ with respect to the $\llex$ ordering (line 2), and sort this set in the increasing $\llex$ order (line 3). We then initialise the main loop of \RPNI with $\PT$ (line 4), and continue until $P$ becomes empty (line 5). Let $D$ be the \dfa at the beginning of an iteration. 
First, pop $(u,v)$ {which is} the minimal element from $P$ (line $6$). Then perform a merge of the states $D(u)$ and $D(v)$ of $D$, \ie, check whether the merged \dfa is consistent with respect to the sample set. If that is indeed (resp.,~not) the case, we accept (resp.,~discard) the merge, and update $D$ accordingly (lines 7-8). Finally, we terminate when $P$ becomes empty, and return the corresponding \dfa.

The procedure \emph{Merge} (line 7) takes as input a \dfa $D$ and two states $D(u)$ and $D(v)$, and returns a new \dfa $D'$ that is obtained from $D$ by merging those two states. We give a high-level idea of the procedure. First, if $D(u) = D(v)$, then $D' = D$. Otherwise, let us denote $D(u)$ and $D(v)$ by $q_u$ and $q_v$, respectively. Also recall that since $(u,v)$ is in the ordered set, hence $v \llex u$. Let $q'$ be such that there exists a transition $q' \xrightarrow{a}q_u$ in $D$ for some $a \in \Sigma$. Then construct $D'$ as follows: delete the state $q_u$ from $D$, redirect all incoming transitions to $q_u$ to $q_v$, and all outgoing transitions from $q_u$ will now be outgoing transitions of $q_v$. To make $D'$ deterministic, we may need to recursively \emph{fold} the subtree of $q_u$ into $q_v$: if $q_u$ and $q_v$ have successors $q_u'$ and $q_v'$ on the same letter $a$, respectively, then merge $q_u'$ and $q_v'$ as well, and repeat this process until there is no common successor remaining.

The termination and correctness of~\cref{alg:rpni} were established in~\cite{oncina1992}. In~\cref{sec:opni}, we will use \RPNI as a black-box in our passive learning algorithm for \drocas.

%% file: opni.tex
\section{OPNI: a passive learning algorithm for DROCA}
\label{sec:opni}

In this section, we introduce \OPNI, a passive learning algorithm for \droca.

\subparagraph*{Learning framework.} 
For passive learning of \drocas, we enhance the input of the learning algorithm with the counter values of prefixes of all words in the sample set. More formally, we assume that along with a consistent sample set $\S = \splus \cup \sminus$ over an alphabet $\Sigma$, the learning algorithm is also provided with the counter-effects for all words in $\pref{\S}$ defined by the function $\ce: \pref{\S}\to \N$.
The learner aims to construct a \droca $\Autom$ that satisfies the following two properties: (1) it accepts all words in $\splus$ and rejects all words in $\sminus$, and (2) for all $w\in\pref{\S}$, $\ce(w)=\ce_{\Autom}(w)$.
In this case, we will say that $\Autom$ is \emph{consistent} with the sample set $\S$.

\subparagraph*{Definitions.}
As in~\cite{MathewPS25}, we define the following two functions. 

\begin{enumerate}
    \item 
Let $\Sigma = \{\sigma_1, \ldots, \sigma_k\}$. 
 Define $\Actions: \Sigma^* \to \{0,1\}\times\{0,1,-1,\dc\}^{k}$ which, given a word $w$, returns a tuple consisting of the sign of the counter value reached after reading $w$ along with the effect on the counter on reading each letter from $\Sigma$ after $w$. 
 Formally, for any word $w$, and $i \in [0,k]$,
\[\Actions(w)[i]=
\begin{cases}
    \sgn(\ce(w)) & \text{if } i=0  \\
     \ce(w \sigma_i)-\ce(w) & \text{if } i>0 \text{ and } \ce(w \sigma_i) \text{ is known}\\
     \dc & \text{if } i>0 \text{ and } \ce(w \sigma_i) \text{ is not known}\\
\end{cases}
\]
Given words $w_1,w_2$, we say that $\Actions(w_1)$ is similar to $\Actions(w_2)$ (denoted $\Actions(w_1) \sim\Actions(w_2)$) if one of the following happens: (a) for all $i \in [1,k]$, we have $\Actions(w_1)[i] = \dc \ \text{or} \ \Actions(w_2)[i] = \dc \text{ or } \Actions(w_1)[i] = \Actions(w_2)[i]$, or (b) the sign of the counter values for words $w_1$ and $w_2$ are different, that is, $\Actions(w_1)[0] \neq \Actions(w_2)[0]$. (As it will be clear later, when $\Actions(w_1) \sim\Actions(w_2)$, we might be able to merge the states reached after reading $w_1$ and $w_2$ respectively.
The merging may be possible if the signs of the counter values are different, regardless of the other components in $\Actions$.) 
Otherwise, $\Actions(w_1)$ and $\Actions(w_2)$ are not similar, and we write $\Actions(w_1) \not \sim\Actions(w_2)$. For $\sigma \in \Sigma$, we write $\Act(w)|_{\sigma}$ to denote the entry in $\Act(w)$ that corresponds to the letter $\sigma$.

\item Define a function that, intuitively, given a word $w$, annotates each letter $w[i]$ of $w$ with the sign of the counter-value reached, upon reading the prefix $w[i-1]$. To that end, for an alphabet $\Sigma$, we define the modified alphabet $\tildeSigma = \bigcup_{\sigma\in\Sigma}\{ \sigma^0, \sigma^1\}$. Then, given a word $w \in \Sigma^+$ and the counter-values $\ce$ of all its prefixes, we define the encryption function as follows: for $i \in [0, |w|-1]$,
 \[\Enc(w)[i]=
\begin{cases}
     w[i]^0 & \text{if } i = 0 \\
     w[i]^{\sgn(\ce(w[0\cdots i-1]))} & \text{if } i>0 \\
\end{cases}
\]
 Additionally, $\Enc(\epsilon)= \epsilon$. 
\end{enumerate}
\begin{example}
    Let $\Sigma = \{a,b\}$ and suppose the following counter-value information are given: $\ce = \{\varepsilon \rightarrow 0, a \rightarrow 1, b \rightarrow 0, ab \rightarrow 0, bb \rightarrow 1\}$. Then $\Actions(a) = (\sgn(\ce(a)), \ce(aa)-\ce(a), \ce(ab)-\ce(a)) = (1, \dc, -1)$. On the other hand, the effect of the $\Enc$ function on the word $ab$, for example, will be $a^0b^1$. 
\end{example}

\subparagraph*{OPNI algorithm.}
We now introduce \OPNI, a passive learning algorithm for \droca, that uses \RPNI as a subroutine. A pseudo-code of the algorithm is given in~\cref{alg:opni}. 
 The algorithm can be broadly divided into two major steps: (1) a preprocessing step that creates an enriched sample set (lines 1-3,  \Cref{alg:opni}), and (2) the inference step that constructs a \droca consistent with the sample set with the help of \RPNI.

\begin{algorithm}[t]
    \caption{\OPNI: a passive learning algorithm for \droca}
    \label{alg:opni}
    \KwInput{An alphabet $\Sigma$, a \emph{consistent} sample set $\S = \splus \cup \sminus$ over $\Sigma$, the function $\ce : \pref{\S} \rightarrow \N$.} 
    \KwOutput{A \droca \( \Autom \) consistent with the sample.}
    \BlankLine
    \emph{Incorporate counter-actions}: Consider the modified sample set $\hatS = \hatsplus \cup \hatsminus$ by replacing every word $u$ in $\S$ with $\Enc(u)$. \\
    \emph{Enriching the alphabet}: Let $\Sigma_{\Act} = \{\Actions(w) \mid w \in \pref{\S}\}$. Then, enrich the alphabet $\Sigma$ with $\Sigma_{\Act}$ as follows: $\widehat{\Sigma} = \widetilde{\Sigma} \cup \Sigma_{\Act}$.
    \label{stepTwo}\\
    \emph{Enriching the sample set}: For every word $w \in \pref{\S}$: add $\Enc(w)\cdot \Actions(w)$ to $\hatsplus$, and for every other $\op \in \Sigma_{\Act}$ such that $\op \not\sim \Actions(w)$, add $\Enc(w)\cdot \op$ to $\hatsminus$.
    \label{stepThree}\\
    \emph{Apply \RPNI on $\hatS$}: Run \RPNI on the sample set $\hatS$ over $\hatSigma$. 
    Let $\hatD \leftarrow$ \RPNI$(\hatS; \hatSigma)$. \\
    {Construct a \droca $\Autom$ from $\hatD$ using the Procedure \hyperlink{alg:constOCA}{\constOCA}}.\\
    \Return \( \Autom \).
    \end{algorithm}

In the preprocessing step, we incorporate the counter information about prefixes of $\S$ to construct an input $\hatS$ for \RPNI. 
 To that end, we first replace every word $u$ in $\S$ with its encoding $\Enc(u)$ (line 1). Let us denote by $\Enc(\S)$ the set $\{\Enc(u) \mid u \in \S\}$.  We then consider the set $\Sigma_{\Act} = \{\Actions(w) \mid w \in \pref{\S}$\}, and  the modified alphabet $\widehat{\Sigma} = \widetilde{\Sigma} \cup \Sigma_{\Act}$ (line 2). Recall that, for an alphabet $\Sigma$, the alphabet $\tildeSigma$ is defined as the set $\bigcup_{\sigma\in\Sigma}\{ \sigma^0, \sigma^1\}$. Finally, for every prefix $w \in \pref{\S}$, we add the word $\Enc(w)\cdot \Actions(w)$ to $\hatsplus$, and for every other $\op \in \Sigma_{\Act}$ such that $\op \not\sim \Actions(w)$, we add the word $\Enc(w)\cdot \op$ to $\hatsminus$ (line 3). 
 
 In the inference step, we run \RPNI on the sample set $\hatS$ over the alphabet $\widehat{\Sigma}$ (line 4). Let $\hatD$ denotes the output \dfa of \RPNI, 
then we construct a \droca $\Autom$ from $\hatD$ using the construction given in Procedure \hyperlink{alg:constOCA}{\constOCA} (line 5).
For the rest of the section, we fix a $\Sigma, \S, \ce$ and $\hatD$.

The following lemmas (\Cref{lem:properties,lem:key-lemma}) ensure that the \dfa constructed by \opni satisfies certain correctness and consistency properties. These properties will later be crucial in proving the correctness of the \opni algorithm.

 \begin{lemma}
 \label{lem:properties}
The \dfa $\hatD$ satisfies the following two properties:
\begin{enumerate}
    \item for all words $w\in \splus$ (resp. $w\in\sminus$), we have $\Enc(w) \in\Lang(\hatD)$ (resp. $\Enc(w)\not\in\Lang(\hatD)$), 
    \item for any two words $w_1,w_2\in\pref{S}$, if the runs on $\Enc(w_1)$ and $\Enc(w_2)$ reach the same state in $\hatD$, then $\Act(w_1) \sim \Act(w_2)$.
\end{enumerate}
 \end{lemma}
 \begin{proof}
    \begin{enumerate}
        \item From line 1 of \cref{alg:opni}, we have that, for all $w \in \splus$ (resp. $w \in \sminus$), $\Enc(w) \in \hatsplus$ (resp. $\Enc(w) \in \hatsminus$). Then the property holds from the correctness of \RPNI.
        \item Assume, towards a contradiction, that there exist words $w_1,w_2\in\S$, such that the runs on $\Enc(w_1)$ and $\Enc(w_2)$ reach the same state in $\hatD$, but $\Act(w_1) \not \sim \Act(w_2)$. From line 3 of \cref{alg:opni}, we have that,
        $\Enc(w_1)\cdot\Act(w_1) \in \hatsplus$, and $\Enc(w_2)\cdot\Act(w_1) \in \hatsminus$.
        Then due to the correctness of \RPNI, it follows that $\Enc(w_1)\cdot\Act(w_1) \in \Lang(\hatD)$, and
        $\Enc(w_2)\cdot\Act(w_1) \not\in \Lang(\hatD)$.
    Now since $\Enc(w_1)$ and $\Enc(w_2)$ reaches the same state in $\hatD$, it must also be the case that, $\Enc(w_2)\cdot\Act(w_1) \in \Lang(\hatD)$, which is a contradiction. \qedhere
    \end{enumerate}%
\end{proof}

\begin{lemma}
    \label{lem:key-lemma}
    For all $q,q' \in Q$ and $\sigma \in \Sigma$ such that $q \xrightarrow{\sigma^0} q'$ (resp., $q \xrightarrow{\sigma^1} q'$) is a transition in $\hatD$, there must exist $\act \in \Sigma_{\Act}$ and $q_f \in Q$ such that $\act[0] = 0$ (resp., $\act[0] = 1$) and $q \xrightarrow{\act}q_f$. Furthermore, $\act|_{\sigma}\in\{0,+1\}$ (resp., $\act|_{\sigma}\in\{-1,0,+1\}$).
\end{lemma}
\begin{proof}
    Let $q,q' \in Q$ and $\sigma \in \Sigma$ be such that $q \xrightarrow{\sigma^0} q'$. Then, there exists $w\sigma \in \pref{\S}$ such that $q_0 \xrightarrow{\Enc(w)} q$ in $\hatD$ with $\ce(w) = 0$ and $\Act(w)[0] = 0 \in \Sigma_{\Act}$. The latter condition implies that $\Enc(w)\cdot\Act(w) \in \hatsplus$. Let $\act = \Act(w)$. Therefore, from the correctness of \RPNI, we have that $\Enc(w)\cdot\act \in \Lang(\hatD)$, which implies that there exists $q_f \in Q$ such that $q \xrightarrow{\act}q_f$ in $\hatD$. Since $w\sigma$ is in $\pref{\S}$, and $\ce(w\sigma)$ is given in the input, we have that $\act|_{\sigma} \in \{0,+1\}$. The proof for the other case, when the counter-effect of $w$ is positive, is similar.
\end{proof}

\subparagraph*{Procedure \constOCA.}
\hypertarget{alg:constOCA}
Let  $\hatD =(Q, \hatSigma, q_0, \delta, F)$.
We then define the $\droca$ $\Autom= (Q, \Sigma, q_0, \delta_0, \delta_1, F)$ where $\delta_0$ and $\delta_1$ are specified as follows.
For every $q\in Q$, $\sigma\in\Sigma$, if there is a transition $q \xrightarrow{\sigma^0}q'$ for some $q'$ in $\hatD$, then, thanks to \cref{lem:key-lemma}, there exist $\act \in \Sigma_{\Act}$ and $q_f \in Q$ such that $\act[0] = 0$ and $q \xrightarrow{\act}q_f$ with $\act|_{\sigma}\in\{0,+1\}$. Let $c = \act|_{\sigma}$. Then define $\delta_0(q,{\sigma}) = (q', c)$. The definition of $\delta_1$ is analogous.

\begin{lemma}
\label{lem:construct-droca}
Let $\Autom = \constOCA(\hatD)$. Then, $\Autom$ is consistent with the sample set $\S$.
\end{lemma}

\begin{proof}
We know, from Lemma~\ref{lem:properties}, that for any word $w_1, w_2 \in\Sigma^*$, if $\hatD$ on reading $\Enc(w_1)$ and $\Enc(w_2)$ reaches the same state, then $\Actions(w_1)$ is similar to $\Actions(w_2)$. This implies that the transition functions of $\Autom$ are well-defined. 
We will now show that $\Autom$ is consistent with the sample $\S$.
\begin{clam}
\label{claim:ce}
    For all $w \in\pref{\S}$, $\ce(w)= \ce_{\Autom}(w)$.
\end{clam}
\begin{clamproof}
We will prove this claim by induction on the length of $w$. This is trivial for $w = \varepsilon$ (base case), since by definition, $\ce(\varepsilon) = \ce_{\Autom}(\varepsilon) = 0$.
Now, assume that the claim is true for all words $u$ of length less than or equal to $l$, for some $l \ge 0$. Let $u\sigma$ be a prefix in $\pref{\S}$. We will show: $\ce(u\sigma) = \ce_{\Autom}(u\sigma)$.
Let $q, q'\in Q$ be such that $\hatD(\Enc(u)) = q$, and $q \xrightarrow{\sigma^c}q'$ is a transition in $\hatD$, where $c = \sgn(\ce(u))$. Then, by the construction of $\Autom$, we have that $\delta_c(q,\sigma) = (q', t)$ where $t = \ce(u\sigma) - \ce(u)$. 
Therefore, $\ce_{\Autom}(u\sigma) = \ce_{\Autom}(u) + t$. Using induction hypothesis, we can rewrite the above equation as  $\ce_{\Autom}(u\sigma) = \ce(u) + t = \ce(u\sigma)$. This concludes the proof of the claim. 
\end{clamproof}
 
\cref{claim:ce} also ensures that $\Enc(w)= \Enc_{\Autom}(w)$. 
By construction of $\Autom$, for all $w \in \pref{\S}$, $\hatD(\Enc(w)) = \Autom(w)$.
Since, from \cref{lem:properties}, for all $w\in \splus$, we have $\Enc(w) \in \Lang(\hatD)$, thus we conclude that $w \in \Lang(\Autom)$. Similarly for $w \in \sminus$, we have $w \not \in \Lang(\Autom)$. Together with \cref{claim:ce}, this concludes the proof of the lemma. 
\end{proof}

\begin{figure}[t]
    \centering
        \begin{subfigure}[t]{0.25\textwidth}
        \centering
        \scalebox{0.5}{\input{./Example_figs/fig1}}
        \caption{\centering Prefix Tree.}
        \label{fig:opni-pta}
    \end{subfigure}
\hfill
    \begin{subfigure}[t]{0.25\textwidth}
    \centering
        \scalebox{0.5}{\input{./Example_figs/fig2}}
        \caption{\centering Merge(\PT;$q_{\varepsilon},q_{a^0}$).}
        \label{fig:opni-step1}
    \end{subfigure}
    \hfill
    \begin{subfigure}[t]{0.2\textwidth}
    \centering
        \scalebox{0.5}{\input{./Example_figs/fig7}}
         \caption{\centering Output of \RPNI.}
        \label{fig:opni-hatD} 
     \end{subfigure}
     \begin{subfigure}[t]{0.2\textwidth}
     \centering
        \scalebox{0.5}{\input{./Example_figs/fig8}}
        \caption{\centering Output \droca.}
        \label{fig:opni-output}
    \end{subfigure}
    \caption{Different steps of the \OPNI algorithm on \cref{eg:opni-example}. \cref{fig:opni-pta} represents the prefix tree \PT; \cref{fig:opni-step1} represents the \dfa after merging $q_{\varepsilon}$ and $q_{a^0}$; \cref{fig:opni-hatD} is the \dfa that is the output of \RPNI; \cref{fig:opni-output} is the \droca output by the \opni algorithm.}
    \label{fig:example}
\end{figure}

Let us now illustrate the \opni algorithm on an example.
\begin{example}
\label{eg:opni-example}
    Let $\S = \splus \cup \sminus$ be a sample set over an alphabet $\Sigma = \{a,b\}$ with $\splus = \{ab, bb\}$ and $\sminus = \{a, b\}$. Suppose the counter-values for words in $\pref{\S}$ are: $\ce = \{\varepsilon \rightarrow 0, a \rightarrow 1, b \rightarrow 0, ab \rightarrow 0, bb \rightarrow 1\}$. Below, we illustrate different steps of \OPNI on this example.
    \begin{itemize}
    \item Using the $\Enc$ function, construct (line 1) the modified sample set $\hatS = \hatsplus \cup \hatsminus$ where $\hatsplus = \{a^0b^1, b^0b^0\}$ and $\hatsminus = \{a^0, b^0\}$ over the alphabet $\hatSigma = \{a^0,a^1,b^0,b^1\}$. 
    \item Then, we construct (line 2) the set $\Sigma_{\Act}$ as follows. We evaluate the $\Actions$ function for words in $\pref{S}$.
    For example, $\Actions(\varepsilon) = (\sgn(\ce(\varepsilon)), \ce(a)-\ce(\varepsilon), \ce(b)-\ce(\varepsilon)) = (0, +1, 0)$. Similarly, we compute $\Actions(w)$ for words in $\pref{\S}$, and get ${\Actions} = \{\varepsilon \rightarrow (0,+1,0), a \rightarrow (1, \dc, -1), ab \rightarrow (0,\dc,\dc), b \rightarrow (0, \dc, +1), bb \rightarrow (1,\dc,\dc)\}$. For simplicity and better readability, we assign symbols from the English alphabet to these action tuples, and write:
    $(0,+1,0) = r, \ (1, \dc, -1) = s, \ (0,\dc,\dc) = t, \ (0, \dc, +1) = u, \text{ and } (1,\dc,\dc)= v$ 
    and therefore, $\Sigma_{\Act} = \{r,s,t,u,v\}$. Consequently, $\hatSigma = \{a^0,a^1,b^0,b^1,r,s,t,u,v\}$.
    \item We then enrich the sample set $\hatS = \hatsplus \cup \hatsminus$ as in line 3 of the algorithm. For example, since $\varepsilon$ is a prefix of $\S$, we will include $\Enc(\varepsilon)\cdot \Act(\varepsilon) = r$ in $\hatsplus$, and since $u$ is not similar to $r$, we include $u$ into $\hatsminus$. Doing this for other elements in $\hatS$, we get the following sets: $\hatsplus = \{a^0b^1, b^0b^0\} \cup \{r, a^0s, a^0b^1t, b^0u, b^0b^0v\}$, and $\hatsminus = \{a^0, b^0\} \cup \{u, b^0r\}$.
    \item Next, we apply \RPNI (line 4) to this newly constructed sample set $\hatS$ over the alphabet $\hatSigma$. Different steps of the \RPNI algorithm on $\hatS$ are shown in \cref{fig:example}. The first two steps are depicted in \cref{fig:opni-pta} and \cref{fig:opni-step1}, and \cref{fig:opni-hatD} represents the output of \RPNI on $\hatS$.  The intermediate steps of \RPNI are shown in \cref{fig:AppExample}.
    \item
    Finally, we construct a \droca $\Autom$ from $\hatD$ using the Procedure \hyperlink{alg:constOCA}{\constOCA} (line 5). For example, consider the transition $q_0 \xrightarrow{b^1} q_2$ in $\hatD$. Due to \cref{lem:key-lemma}, there must exist a transition from its source state $q_0$ on an action $\act$ with $\act[0] = 1$. We deduce that $\act = r = (1, \dc, -1)$, and the transition is $q_0 \xrightarrow{r} q_2$. Therefore, according to \hyperlink{alg:constOCA}{\constOCA} procedure, we assign $\delta_1(q_0,b) = (q_2,-1)$. \cref{fig:opni-output} represents the output \droca of \hyperlink{alg:constOCA}{\constOCA} on $\hatD$. 
    \end{itemize}
\end{example}

\begin{figure}[h]
    \centering
        \begin{subfigure}[t]{0.45\textwidth}
        \centering
        \scalebox{.8}{\input{./Example_figs/fig3}}
        \caption{\centering Merge of $a^0b^1,r$.}
        \label{fig:step-2}
    \end{subfigure}
    \begin{subfigure}[t]{0.45\textwidth}
    \centering
        \scalebox{.8}{\input{./Example_figs/fig4}}
        \caption{\centering Merge of $a^0s,r$.}
        \label{fig:step-3}
    \end{subfigure}
    \begin{subfigure}[t]{0.45\textwidth}
    \centering
        \scalebox{.8}{\input{./Example_figs/fig5}}
         \caption{\centering Merge of $b^0b^0,r$.}
        \label{fig:step-4} 
     \end{subfigure}
     \begin{subfigure}[t]{0.5\textwidth}
     \centering
        \scalebox{.8}{\input{./Example_figs/fig6}}
        \caption{\centering Merge of $b^0u,r$.}
        \label{fig:step-5}
    \end{subfigure}
    \caption{Intermediate steps of the \OPNI algorithm on \cref{eg:opni-example}.}
    \label{fig:AppExample}
\end{figure}

\subparagraph*{Justification of counter values in the input.} We make use of the counter information in the preprocessing step to construct two sets: $\Enc(\S)$ and $\Sigma_{\Act}$. The counter-actions and signs of counter values in the input are required to prevent \RPNI on the sample set $\hatS$ (line 4) from doing some of the merges. Let us illustrate this with an example. 
Suppose in a certain iteration of \RPNI, we are trying to merge two states $D(u)$ and $D(v)$, for some prefixes $u$ and $v$ of $\splus$, and let both states have outgoing transitions on a letter $a$. As such, if there is no information on the counter-actions, \RPNI will merge these two states, provided that the merged automaton is consistent with $\sminus$. 
However, suppose that $\sgn(\ce(u)) = \sgn(\ce(v))$, but $\Actions(u) \not\sim \Actions(v)$. In that case, these two states should not be merged, since in any \droca that is consistent with the sample, $u$ and $v$ must lead to two different states. This justifies the construction of $\Sigma_{\Act}$. Now suppose that $\sgn(\ce(u)) \neq \sgn(\ce(v))$. In that case, those two states could very well be merged. This justifies the preprocessing steps in \OPNI.

 Since \RPNI is guaranteed to terminate, together with \cref{lem:construct-droca}, we conclude the following.
\begin{theorem}
    \cref{alg:opni} terminates, and returns a \droca consistent with the input.
\end{theorem}

%% file: Example_figs/fig1.tex


\begin{tikzpicture}[
  vertex/.style = {circle, draw, minimum size=8mm, inner sep=1pt},
  final/.style = {double, double distance=1pt},
  ->, >=Stealth
  ]
\Large
\node[vertex] (1) at (0,1) {}
  child[grow=down] {node[vertex] (2) at (-2,-0.2) {}
    child {node[vertex, final] (5) at (-1.4,-0.4) {}}
    child {node[vertex, final] (6) at (-.8,-0.4) {}
      child {node[vertex, final] (7) at (0,-0.6) {}}
    }
  }
  child[grow=down] {node[vertex] (3) at (0,-0.2) {}
    child {node[vertex, final] (8) at (.8,-0.4) {}
      child {node[vertex, final] (10) at (0,-0.6) {}}
    }
    child {node[vertex, final] (9) at (1.4,-0.4) {}}
  }
  child[grow=down] {node[vertex, final] (4) at (2,-0.2) {}};
  
\node (init) at (0,2) {};
\draw[->] (init) -- (1);

\draw[->] (1) -- (2) node[midway, left] {\huge $a^0$};
\draw[->] (1) -- (3) node[midway, right] {\huge $b^0$};
\draw[->] (1) -- (4) node[midway, above right] {\huge $r$};

\draw[->] (2) -- (5) node[midway, left] {\huge $s$};
\draw[->] (2) -- (6) node[midway, right] {\huge $b^1$};

\draw[->] (6) -- (7) node[midway, right] {\huge $t$};

\draw[->] (3) -- (8) node[midway, left] {\huge $b^0$};
\draw[->] (3) -- (9) node[midway, right] {\huge $u$};

\draw[->] (8) -- (10) node[midway, right] {\huge $v$};

\end{tikzpicture}


%% file: Example_figs/fig2.tex


\begin{tikzpicture}[
 vertex/.style = {circle, draw, minimum size=8mm, inner sep=1pt},
  final/.style = {double, double distance=1pt},
  ->, >=Stealth
  ]
\Large
\node[vertex] (1) at (0,0) {};

\node (init) at (0,1) {};
\draw[->] (init) -- (1);

\draw[->] (1) edge[loop left] node {\huge $a^0$} (1);

\node[vertex] (3) at (-2.5,-2) {};
\node[vertex, final] (4) at (-.7,-2) {};
\node[vertex, final] (5) at (.7,-2) {};
\node[vertex, final] (6) at (2.5,-2) {};

\node[vertex, final] (8) at (-2.5,-4) {};
\node[vertex, final] (9) at (-1,-4) {};
\node[vertex, final] (7) at (2.5,-4) {};

\node[vertex, final] (10) at (-2.5,-6) {};

\draw[->] (1) -- (3) node[midway, left] {\huge $b^0$};
\draw[->] (1) -- (4) node[midway, left] {\huge $r$};
\draw[->] (1) -- (5) node[midway, right] {\huge $s$};
\draw[->] (1) -- (6) node[midway, right, xshift=0.2cm] {\huge $b^1$};

\draw[->] (3) -- (8) node[midway, left] {\huge $b^0$};
\draw[->] (3) -- (9) node[midway, right] {\huge $u$};

\draw[->] (8) -- (10) node[midway, right] {\huge $v$};

\draw[->] (6) -- (7) node[midway, right] {\huge $t$};

\end{tikzpicture}


%% file: Example_figs/fig7.tex


\begin{tikzpicture}[
  vertex/.style = {circle, draw, minimum size=8mm, inner sep=1pt},
  final/.style = {double, double distance=1pt},
  ->, >=Stealth
  ]
\Large
\node[vertex] (1) at (0,0) {\huge $q_0$};

\node (init) at (0,1) {};
\draw[->] (init) -- (1);

\draw[->] (1) edge[loop left] node {\huge $a^0$} (1);

\node[vertex] (3) at (-1.2,-2) {\huge $q_1$};
\node[vertex, final] (4) at (1.2,-2) {\huge $q_2$}; 

\draw[->] (1) -- (3) node[midway, left] {\huge $b^0$};
\draw[->] (1) -- (4) node[midway, right] {\huge $b^1, r, s$};

\draw[->] (3) -- (4) node[midway, below] {\huge $b^0, u$};

\draw[->] (4) edge[loop right] node {\huge $t, v$} (4);

  \draw[dashed, gray!10] (2,-4.2) -- (2,-4.2);
\end{tikzpicture}


%% file: Example_figs/fig8.tex


\begin{tikzpicture}[
  vertex/.style = {circle, draw, minimum size=8mm, inner sep=1pt},
  final/.style = {double, double distance=1pt},
  ->, >=Stealth
  ]
\Large
\node[vertex] (1) at (0,0) {\huge $q_0$};

\node (init) at (0,1) {};
\draw[->] (init) -- (1);

\draw[->] (1) edge[loop left] node {\huge $a_{=0}, +1$} (1);

\node[vertex] (3) at (-1.2,-2) {\huge $q_1$};
\node[vertex, final] (4) at (1.2,-2) {\huge $q_2$};

\draw[->] (1) -- (3) node[midway, left] {\huge $b_{=0}, 0$};
\draw[->] (1) -- (4) node[midway, right] {\huge $b_{>0}, -1$};

\draw[->] (3) -- (4) node[midway, below,yshift=-.4cm] {\huge $b_{=0}, +1$};
  \draw[dashed, gray!10] (2,-4.2) -- (2,-4.2);
\end{tikzpicture}


%% file: Example_figs/fig3.tex


\begin{tikzpicture}[
  vertex/.style = {circle, draw, minimum size=8mm, inner sep=1pt},
  final/.style = {double, double distance=1pt},
  ->, >=Stealth
  ]

\node[vertex] (1) at (0,0) {};

\node (init) at (0,1) {};
\draw[->] (init) -- (1);

\draw[->] (1) edge[loop left] node {\huge $a^0$} (1);

\node[vertex] (3) at (-3,-2.5) {};
\node[vertex, final] (4) at (0,-2.5) {}; 
\node[vertex, final] (5) at (1.2,-2.5) {}; 

\node[vertex, final] (8) at (-3.5,-4) {};
\node[vertex, final] (9) at (-2.5,-4) {};
\node[vertex, final] (7) at (0,-4) {}; 

\node[vertex, final] (10) at (-3.5,-5.3) {}; 

\draw[->] (1) -- (3) node[midway, left] {\huge $b^0$};
\draw[->] (1) -- (4) node[pos=.9, left] {\huge $b^1, r$};
\draw[->] (1) -- (5) node[midway, right] {\huge $s$};

\draw[->] (3) -- (8) node[midway, left] {\huge $b^0$};
\draw[->] (3) -- (9) node[midway, right] {\huge $u$};

\draw[->] (8) -- (10) node[midway, right] {\huge $v$};

\draw[->] (4) -- (7) node[midway, right] {\huge $t$};

\end{tikzpicture}


%% file: Example_figs/fig4.tex


\begin{tikzpicture}[
  vertex/.style = {circle, draw, minimum size=8mm, inner sep=1pt},
  final/.style = {double, double distance=1pt},
  ->, >=Stealth
  ]

\node[vertex] (1) at (0,0) {};

\node (init) at (0,1) {};
\draw[->] (init) -- (1);

\draw[->] (1) edge[loop left] node {\huge $a^0$} (1);

\node[vertex] (3) at (-1.2,-2.5) {}; 
\node[vertex, final] (4) at (1.2,-2.5) {}; 

\node[vertex, final] (8) at (-1.7,-4) {};
\node[vertex, final] (9) at (-0.7,-4) {};
\node[vertex, final] (7) at (1.2,-4) {}; 

\node[vertex, final] (10) at (-1.7,-5.3) {}; 

\draw[->] (1) -- (3) node[midway, left] {\huge $b^0$};
\draw[->] (1) -- (4) node[midway, right] {\huge $b^1, r, s$};

\draw[->] (3) -- (8) node[midway, left] {\huge $b^0$};
\draw[->] (3) -- (9) node[midway, right] {\huge $u$};

\draw[->] (8) -- (10) node[midway, right] {\huge $v$};

\draw[->] (4) -- (7) node[midway, right] {\huge $t$};

\end{tikzpicture}


%% file: Example_figs/fig5.tex


\begin{tikzpicture}[
  vertex/.style = {circle, draw, minimum size=8mm, inner sep=1pt},
  final/.style = {double, double distance=1pt},
  ->, >=Stealth
  ]

\node[vertex] (1) at (0,0) {};

\node (init) at (0,1) {};
\draw[->] (init) -- (1);

\draw[->] (1) edge[loop left] node {\huge $a^0$} (1);

\node[vertex] (3) at (-1.2,-2.5) {};
\node[vertex, final] (4) at (1.2,-2.5) {}; 

\node[vertex, final] (9) at (-1.3,-4) {};  
\node[vertex, final] (7) at (0,-4) {};   
\node[vertex, final] (10) at (1.3,-4) {}; 

\draw[->] (1) -- (3) node[midway, left] {\huge $b^0$};
\draw[->] (1) -- (4) node[midway, right] {\huge $b^1, r, s$};

\draw[->] (3) -- (9) node[midway, left] {\huge $u$};
\draw[->] (3) -- (4) node[midway, above] {\huge $b^0$}; 

\draw[->] (4) -- (7) node[midway, left] {\huge $t$};
\draw[->] (4) -- (10) node[midway, right] {\huge $v$};

 \draw[dashed, gray!10] (2,-4.5) -- (2,-4.5);

\end{tikzpicture}


%% file: Example_figs/fig6.tex


\begin{tikzpicture}[
  vertex/.style = {circle, draw, minimum size=8mm, inner sep=1pt},
  final/.style = {double, double distance=1pt},
  ->, >=Stealth
  ]

\node[vertex] (1) at (0,0) {};

\node (init) at (0,1) {};
\draw[->] (init) -- (1);

\draw[->] (1) edge[loop left] node {\huge $a^0$} (1);

\node[vertex] (3) at (-1.2,-2.5) {};
\node[vertex, final] (4) at (1.2,-2.5) {}; 

\node[vertex, final] (7) at (0,-4) {};   
\node[vertex, final] (10) at (1.2,-4) {}; 

\draw[->] (1) -- (3) node[midway, left] {\huge $b^0$};
\draw[->] (1) -- (4) node[midway, right] {\huge $b^1, r, s$};

\draw[->] (3) -- (4) node[midway, above] {\huge $b^0,u$}; 

\draw[->] (4) -- (7) node[midway, left] {\huge $t$};
\draw[->] (4) -- (10) node[midway, right] {\huge $v$};

 \draw[dashed, gray!10] (2,-4.5) -- (2,-4.5);
 
\end{tikzpicture}


%% file: active-passive.tex
\section{\opniL: \OPNI-guided active learning procedure for OCA}
\label{sec:active-passive}

In this section, we provide a procedure for active learning of \droca using \opni. 

\subparagraph*{Learning framework.} 
Following~\cite{MathewPS25, bps}, in an \emph{active learning} framework for \drocas, the \emph{learner}'s aim is to construct a \droca $\Butom$ by interacting with a \emph{teacher} that has a \droca $\Autom$ in mind, such that $\Autom$ and $\Butom$ are equivalent. In the process, the learner can ask the teacher three types of queries:
\begin{itemize}
\item \emph{Membership queries} \textsf{MQ}$_\Autom$: the learner provides a  word $w\in\Sigma^*$. The teacher returns $1$ if  $w\in\Lang(\Autom)$, and $0$ if $w\not\in\Lang(\Autom)$.
\item \emph{Counter value queries} \textsf{CV}$_\Autom$: the learner asks the counter value reached on reading a word $w$. The teacher returns the counter value, \ie, $\ce_{\Autom}(w)$.
\item \emph{Minimal synchronous-equivalence queries} \textsf{MSQ}$_\Autom$: the learner asks whether a \droca $\Cutom$ is equivalent and counter-synchronous to $\Autom$. The teacher returns \emph{yes} only if $\Cutom$ and $\Autom$ are counter-synchronous and equivalent. Otherwise, the teacher provides a `minimal' counter-example {$z \in \Sigma^*$} such that either $\Cutom(z) \neq \Autom(z)$ or $\ce_{\Autom}(z) \neq \ce_{\Cutom}(z)$. 
\end{itemize}

Note that the minimal synchronous-equivalence queries (introduced in \cite{MathewPS25}) may return a word for which the two machines reach different counter values. 
This query enables the active learning method to construct a \droca that is counter-synchronous and equivalent to the teacher's \droca. As noted in \cite{MathewPS25},  synchronous equivalence is significantly faster than the impractical general equivalence check. Hence, relaxing this condition and having minimal equivalence queries will require the teacher to use a general equivalence check. Thus, in our implementation (described in \cref{sec:experiments}), we opt for equivalence queries on minimal counter-synchronous \droca, as they are faster. 

\paragraph*{Learning of DROCA}
We introduce an active learning procedure for \droca, called \opniL.
The pseudocode is presented in Procedure~\ref{alg:activeLearn}, and its key components are detailed below. \opniL builds on Angluin's $L^*$ algorithm with modifications inspired by the \minOCA framework from~\cite{MathewPS25}.

\subparagraph*{Observation table.}
The learner maintains an observation table over the input alphabet $\Sigma= \{\sigma_1,\ldots, \sigma_k\}$ for some $k\in\N$.
An observation table $T$ is given by $T=(\R,\C,\Memb,\ce\upharpoonright_{\R\cup\R\Sigma}, \Act)$, where $\R$ is a nonempty prefix-closed set of strings, $\C$ is a nonempty suffix-closed set of string, $\Memb:(\R \cup \R \Sigma)  \C \to \{0, 1\}$ is a function that indicates whether a word belongs to the language, $\ce\upharpoonright_{\R\cup\R\Sigma}: \R\cup\R\Sigma \to \N$ is the function $\ce$ with domain restricted to the set $\R\cup\R\Sigma$, and $\Actions: (\R \cup \R \Sigma)  \C \to \{0,1\}\times\{0,1,-1 \}^k$ is a function representing the sign of the counter value reached and the counter-actions on every letter after reading a word. 
The function $\Act$ is the same as the one defined in \cref{sec:opni}. 
Given $w\in (\R\cup \R  \Sigma)  \C$, the function $\Memb(w)$ returns $0$ (resp., $1$) if $\Autom(w)$ is equal to $0$ (resp., $1$).
The observation table initially has $\R=\C=\{\epsilon\}$ and is updated as the procedure runs.
Let $\C=\{c_1,\ldots,c_{\ell}\}$ for some $\ell\in\N$.
For any $r\in\R\cup\R\Sigma$, we use $row(r)$ to denote the tuple $(\ce(r), (\Memb(rc_1),\Act(rc_1)), \ldots, (\Memb(rc_\ell),\Act(rc_\ell)))$. 

\begin{definition}[Definition 5 in \cite{MathewPS25}]
Let $d\in\N$ and $T = (\R,\C,\Memb,\ce\upharpoonright_{\R\cup\R\Sigma},\Act)$ be an observation table. 
\begin{enumerate}
\item  $T$ is said to be \emph{$d$-closed} if for all $r'\in \R\Sigma$ with $\ce(r')\leq d$ there exists $r\in\R$ such that $row(r) = row(r^\prime)$. 
\item  $T$ is said to be \emph{$d$-consistent} if for all $r,s\in\R$, such that $\ce(r)=\ce(s) \leq d$ and $row(r)=row(s)$ implies that for all $\sigma\in \Sigma$, $row(r\sigma)=row(s\sigma)$. 
\end{enumerate}
\end{definition}

As described in Procedure~\ref{alg:activeLearn} (see also~\cite{MathewPS25}), the observation table is iteratively expanded until it satisfies the conditions of $d$-closure and $d$-consistency. The set of words $w \in (\R \cup \R \Sigma)\C$ is partitioned into $\splus$ and $\sminus$, based on whether each word is accepted or rejected by the teacher. The goal is to construct a \droca\ that solves the corresponding passive learning task: separating $\splus$ from $\sminus$ while being counter-synchronous with the teacher's \droca\ for all prefixes of words in $\splus \cup \sminus$. This task is delegated to the passive learning algorithm \opni, which returns a \droca\ consistent with the observed data. 
A \droca and an observation table corresponding to it are given in \Cref{ex:droca} and \Cref{obtable}, respectively. \begin{figure}[h!]
\begin{minipage}[b]{0.48\textwidth}
\centering
\scalebox{1}{
\begin{tikzpicture}[shorten >=1pt,node distance=3cm,on grid,auto]
\tikzset{every path/.style={line width=.4mm}}
   \tikzset{initial text={}}
   \node[state, initial] (q0) at (0,0) {$q_0$};
   \node[state]          (q1) at (3,0) {$q_1$};
   \node[state, accepting] (q2) at (3,-2) {$q_2$};

   \path[->]
   (q0) edge[loop above] node {$a_{=0}, +1;\ a_{>0}, +1$} (q0)
        edge[left] node[xshift=-.2cm] {$b_{=0}, 0$} (q2)
        edge[above] node {$b_{>0}, -1$} (q1)
   (q1) edge[loop above] node {$b_{>0}, -1$} (q1)
        edge[right] node {$b_{=0}, 0$} (q2)
   (q2) edge[loop below] node {$b_{=0}, 0$} (q2);
\end{tikzpicture}
}
\caption{A \textsc{droca} recognising the language $\{a^n b^m \mid m > n\}$. Transitions not shown go to a non-final sink state with counter-action $0$.}
\label{ex:droca}
\end{minipage}
\hfill
\begin{minipage}[b]{0.48\textwidth}
\centering
\scalebox{1}{
\begin{tabular}{l|c|c|cc|}
                     &                   & $\mathsf{ce}$ & \multicolumn{2}{c|}{$\epsilon$}       \\ \cline{4-5} 
                     &                   &               & \multicolumn{1}{c|}{Memb} & $\mathsf{Act}$ \\ \Xhline{1pt} 
\multirow{4}{*}{\rotatebox{90}{$\mathsf{R}$}}   
                     & $\epsilon$        & 0             & \multicolumn{1}{c|}{0}   & $(0,+1,0)$     \\ \cline{2-5} 
                     & a                 & 1             & \multicolumn{1}{c|}{0}   & $(1,+1,-1)$    \\ \cline{2-5} 
                     & b                 & 0             & \multicolumn{1}{c|}{1}   & $(0,0,0)$      \\ 
                     \cline{2-5} 
                     & ab                & 0             & \multicolumn{1}{c|}{0}   & $(0,0,0)$      \\ 
                     \Xhline{1pt}
\multirow{5}{*}{\rotatebox{90}{$\mathsf{R}\Sigma$}}                      
                     & aa                & 2             & \multicolumn{1}{c|}{0}   & $(1,+1,-1)$    \\ \cline{2-5} 
                     & ba                & 0             & \multicolumn{1}{c|}{0}   & $(0,0,0)$      \\ \cline{2-5} 
                     & bb                & 0             & \multicolumn{1}{c|}{1}   & $(0,0,0)$      \\
                     \cline{2-5} 
                     & aba                & 0             & \multicolumn{1}{c|}{0}   & $(0,0,0)$      \\
                     \cline{2-5} 
                     & abb                & 0             & \multicolumn{1}{c|}{1}   & $(0,0,0)$      \\
                     \cline{2-5} 
                     
\hline
\end{tabular}
}
\caption{Observation table for the \textsc{droca} in \Cref{ex:droca}. The table is $1$-closed and $1$-consistent. Here, $\mathsf{R} = \{\epsilon, a, b\}$ and $\mathsf{C} = \{\epsilon\}$.}
\label{obtable}
\end{minipage}
\end{figure}

\begin{figure}[h!]
\begin{minipage}[b]{0.45\textwidth}
\centering
\scalebox{1}{
\begin{tikzpicture}[shorten >=1pt,node distance=3cm,on grid,auto]
\tikzset{every path/.style={line width=.4mm}}
   \tikzset{initial text={}}
   \node[state, initial] (q0) at (0,0) {$q_0$};
   \node[state, accepting] (q1) at (3,0) {$q_1$};

   \path[->]
   (q0) edge[loop above] node {$a_{=0}; a_{>0}$} ()
        edge[above] node {$b_{>0}$} (q1)
   (q1) edge[loop above] node {$b_{>0}$} ();
\end{tikzpicture}
}
\caption{A \textsc{voca} over the alphabet $\{a,b\}$, where $\Sigma_{call}=\{a\}$, $\Sigma_{ret}=\{b\}$ and $\Sigma_{int}=\emptyset$. The \textsc{voca} recognises the language $\{a^n b^m \mid m \leq n\}$. The transitions not shown go to a non-final sink state.}
\label{ex:voca}
\end{minipage}
\hfill
\begin{minipage}[b]{0.45\textwidth}
\centering
\scalebox{1}{
\begin{tabular}{l|c|c|c|}
         &                  & $\mathsf{ce}$ & $\epsilon$ \\ \cline{4-4}
         &                  &               & Memb       \\ \Xhline{1pt}
\multirow{3}{*}{\rotatebox{90}{$\mathsf{R}$}}   
         & $\epsilon$       & 0             & 0          \\ \cline{2-4} 
         & a                & 1             & 0          \\ \cline{2-4} 
         & ab               & 0             & 1          \\ \Xhline{1pt}
\multirow{4}{*}{\rotatebox{90}{$\mathsf{R}\Sigma$}}                      
         & b                & \texttt{x}    & \texttt{x} \\ \cline{2-4} 
         & aa               & 2             & 0          \\ \cline{2-4} 
         & aba              & 1             & 0          \\ \cline{2-4} 
         & abb              & \texttt{x}    & \texttt{x} \\
\hline
\end{tabular}
}
\caption{Observation table for the \textsc{voca} in \Cref{ex:voca}. The table is both $1$-closed and $1$-consistent. Here, $\mathsf{R}=\{\epsilon, a, ab\}$ and $\mathsf{C}=\{\epsilon\}$.}
\label{obtableVoca}
\end{minipage}
\end{figure}

The learner then queries the teacher to verify whether the hypothesis \droca\ accepts the same language as the target. If not, the teacher provides a counterexample, which is used to update the observation table. This process repeats until the learner synthesises a \droca\ that is equivalent to the teacher’s.

It is instructive to contrast our approach with the \minOCA algorithm~\cite{MathewPS25}. While \minOCA uses a SAT solver to solve the passive learning problem and guarantees a minimal separating \droca, the use of \opni does not provide any bound on the size of the returned automaton. As a result, \opni may produce increasingly large automata, and the overall learning procedure \opniL may, in principle, fail to terminate. On the other hand, reliance on SAT solvers limits the scalability of \minOCA. In the next section, we show that our approach, based on \opni, significantly improves scalability compared to \minOCA.

\renewcommand{\algorithmcfname}{Procedure}
\begin{algorithm}[t]
\caption{\opniL for \droca}
\label{alg:activeLearn}
\KwIn{Access to membership (\textsf{MQ}$_\Autom$), counter value (\textsf{CV}$_\Autom$), and  equivalence (\textsf{MSQ}$_\Autom$) queries.}
\KwOut{A \droca $\Butom$ accepting the same language as $\Autom$.}

Initialise $\R \gets \{\varepsilon\}$, $\C \gets \{\varepsilon\}$, $d \gets 0$.\\
Initialise the observation table $T=(\R,\C,\Memb,\ce\upharpoonright_{\R\cup\R\Sigma}, \Act)$.

\Repeat{teacher replies yes to an equivalence query}{
  \While{$T$ is not $d$‐closed or not $d$‐consistent}{
    \If{$T$ is not $d$‐closed}{
     Find $r\in \R$, $a\in\Sigma$ such that $\ce(r a)\leq d$, $row(r a) \neq row(r^\prime)$ for all $r^\prime \in \R$.\\ 
	Add $r a$ to $\R$.	\label{dclosed}
    }
    \If{$T$ is not $d$‐consistent}{
    Find $r,s\in \R$, $a\in\Sigma$, $c\in \C$ such that $\ce(r)=\ce(s) \leq d$, $row(r)= row(s)$, and   
    {($\Memb(r a  c) \neq \Memb(s a  c)$ or $\Actions(r  a  c) \neq \Actions(s  a  c)$).}\\
    Add $a  c$ to $\C$. \label{dconsistent}
    }
    Extend $\Memb$ and $\Act$ to $(\R\cup\R\Sigma)\C$, using membership and counter value queries.
  }
  $\splus = \{w \in (\R\cup\R\Sigma)\C \mid \Memb(w) = 1\}$.\\
  $\sminus = \{w \in (\R\cup\R\Sigma)\C \mid \Memb(w) = 0\}$.\\
  $\Butom \gets \OPNI(\splus \cup \sminus, \ce\upharpoonright_{(\R\cup\R\Sigma)\C}; \Sigma)$.\\
  Ask equivalence query $\textsf{MSQ}_{\Autom}(\Butom)$.

  \If{teacher gives counter‐example $z$ }{
    Add all prefixes of $z$ to $\R$.\\
    Update $d \gets \max\{\max(d, \ce_{\Autom}(z'))\mid z'$ is a prefix of $z\}$.
  }
}
\Return{$\Butom$}.
\end{algorithm}

\paragraph*{Learning of VOCA}
We extend the $\opniL$ method for learning \voca as well. In this setting, the counter-actions are determined by the input alphabet and can therefore be inferred directly from the word. Consequently, the function $\Act$ need not be included in the observation table. This simplifies the learning process, as the sets $\splus$ and $\sminus$ constructed by $\opniL$ are considerably smaller.
Furthermore, there are words that do not have a valid run and are denoted by $\dc$ in the observation table. The notion of $d$-closed and $d$-consistent has to be modified such that a $\dc$-marked cell need not be compared with a cell in the same column. Moreover, the equivalence check for \voca is more efficient than for \droca. A \voca and an observation table corresponding to it are given in \Cref{ex:voca} and \Cref{obtableVoca} respectively (see \Cref{app:Observation}).

%% file: experiments.tex
\section{Experimental evaluation}
\label{sec:experiments}

{Bruy{\`{e}}re et al.~\cite{bps}, were the first to provide a tool for learning \droca. They employ counter value queries in addition to standard membership and equivalence queries. While effective in principle, this approach does not scale well as the size of the \droca being learnt increases beyond $7$ states.

To address these limitations, Mathew et al. \cite{MathewPS25} introduced \minOCA, a learning tool capable of inferring \droca with at most $15$ states using a SAT-based approach. Although \minOCA improves upon earlier implementations, its reliance on SAT solvers leads to a significant increase in computational overhead as the size of the automaton grows. In practice, this results in a bottleneck that makes learning larger \minOCA infeasible.

In this work, we present a new approach— \opniL —which replaces the SAT-based inference in \minOCA with \opni (see \Cref{sec:opni}). This scales effectively with input size and demonstrates the ability to learn larger \droca, significantly outperforming existing tools.}

We implemented \opniL in Python and tested it on randomly generated \droca. We also implemented a faster learning algorithm for the special case of \voca. The implementation details and the results obtained are discussed in this section.

\subparagraph*{Equivalence query.}  The equivalence of \droca is known to be in polynomial time~\cite{droca}. However, as pointed out in \cite{MathewPS25, bps}, the polynomials involved are not suitable for practical applications. Mathew et al.~\cite{MathewPS25} give a practical algorithm for checking the equivalence of \voca and counter-synchronous \droca that returns the minimal counter-example. We use their ideas in checking the equivalence of the learnt automaton in our implementation.

\subparagraph{Generation of random benchmarks.}
We generated two datasets for evaluating the performance of the proposed method.
\begin{itemize}
\item \dsOne: Random \droca to compare the performance of \opniL with \minOCA.
\item \dsTwo: Random \voca to evaluate the performance of \opniL for learning \voca and a similar adaptation of the \minOCA algorithm.
\end{itemize}

\textit{1. Generating random \DROCA (\dsOne).}
{In~\cite{MathewPS25}, experiments were carried out on randomly generated \drocas.
We follow the same procedure to compare our experimental results with those of~\cite{MathewPS25}.}

{The procedure is explained here for the sake of completeness. 
Let $n\in\N$ be the number of states of the \droca to be generated. 
First, we initialise the set of states $Q=\{q_1,q_2,\ldots, q_n\}$. For all $q\in Q$, we add $q$ to the set of final states $F$ with probability 0.5. 
If $Q=F$ or $F=\emptyset$ after this step, then we restart the procedure. 
Otherwise, for all $q\in Q$ and $a\in\Sigma$, we assign $\delta_0(q,a)=(p,c)$ (resp. $\delta_1 (q,a)=(p,c)$), with $p$ a random state in $Q$ and $c$ a random counter operation in $\{0,+1\}$ (resp. $\{0,+1,-1\}$). 
The constructed \droca { is} $\Autom=(Q, \Sigma, \{q_1\}, \delta_0,\delta_1, F)$. 
If the number of reachable states of $\Autom$ from the initial configuration is not $n$, then we discard $\Autom$ and restart the whole procedure. 
Otherwise, we output $\Autom$.}

We generated random \droca with number of states ranging from $2$ to $30$ and alphabet size ranging from $2$ to $5$.
A total of $100$ \drocas were generated for {every pair of alphabet and number of states.}
We reused the dataset used in \cite{minOCA, MathewPS25} for \droca up to $15$ states. 

\textit{2. Generating random \VOCA (\dsTwo).}
For the case of \voca, we use a similar procedure for generating random \voca. 
We generate the states in the same way as above.
Following this, for all $a\in\Sigma$, we {randomly} puts it in one of the sets $\Sigma_{call}, \Sigma_{ret}$, or $\Sigma_{init}$. 
The elements in 
$\Sigma_{call}, \Sigma_{ret}$ or $\Sigma_{init}$, will respectively have $+1,-1$ and $0$ as their corresponding counter action.
If there does not exist at least one symbol in $\Sigma_{call}$ and one in $\Sigma_{ret}$, then we restart the whole procedure, since the language recognised will be regular. 
The transitions are also generated as in the case of \droca.
Note that the counter action $c$ in a transition depends on whether the symbol $a$ is in $\Sigma_{call}, \Sigma_{ret}$, or $\Sigma_{init}$.
The constructed \voca { is} $\Autom=(Q, \Sigma_{call}\cup \Sigma_{ret}\cup \Sigma_{init}, \{q_1\}, \delta_0,\delta_1, F)$.
If the number of reachable states of \Autom from the initial configuration is not $n$, then we discard $\Autom$ and restart the whole procedure. Otherwise, we output $\Autom$.

We generated random \voca with number of states ranging from $10$ to $100$ and input alphabet size ranging from $2$ to $5$. 
A total of $10$ \voca were generated for {every pair of input alphabet and number of states.} {The number of randomly generated \voca was limited to 10 per configuration due to the increasing difficulty in constructing a valid automaton where all states are reachable and satisfy the remaining conditions. Moreover, our primary objective is to evaluate how well the proposed approach scales with input size. For this purpose, the selected sample size is sufficient to observe meaningful trends in performance without incurring excessive computational overhead. }

\subparagraph*{Experimental results.}
All the experiments were performed on an Apple M1 chip with 8GB of RAM.
We implemented the proposed method (\opniL) in Python for \droca and then for the special case of \voca. 
In \Cref{successF} (resp. \Cref{successV}), the total number of languages is out of 400 (resp. 40) for each input size because, for every number of states, we generated 400 (resp. 40) random \droca (resp. \voca) for input alphabet sizes 2, 3, 4, and 5; to keep the visualization simple, we omitted the z-axis for alphabet size and instead aggregated the number of successfully learnt languages across all alphabet sizes for each input size—this approach is used consistently across all graphs.

\begin{figure}[t]
    \centering
        \begin{subfigure}[t]{0.45\textwidth}
        \includegraphics[width=\linewidth]{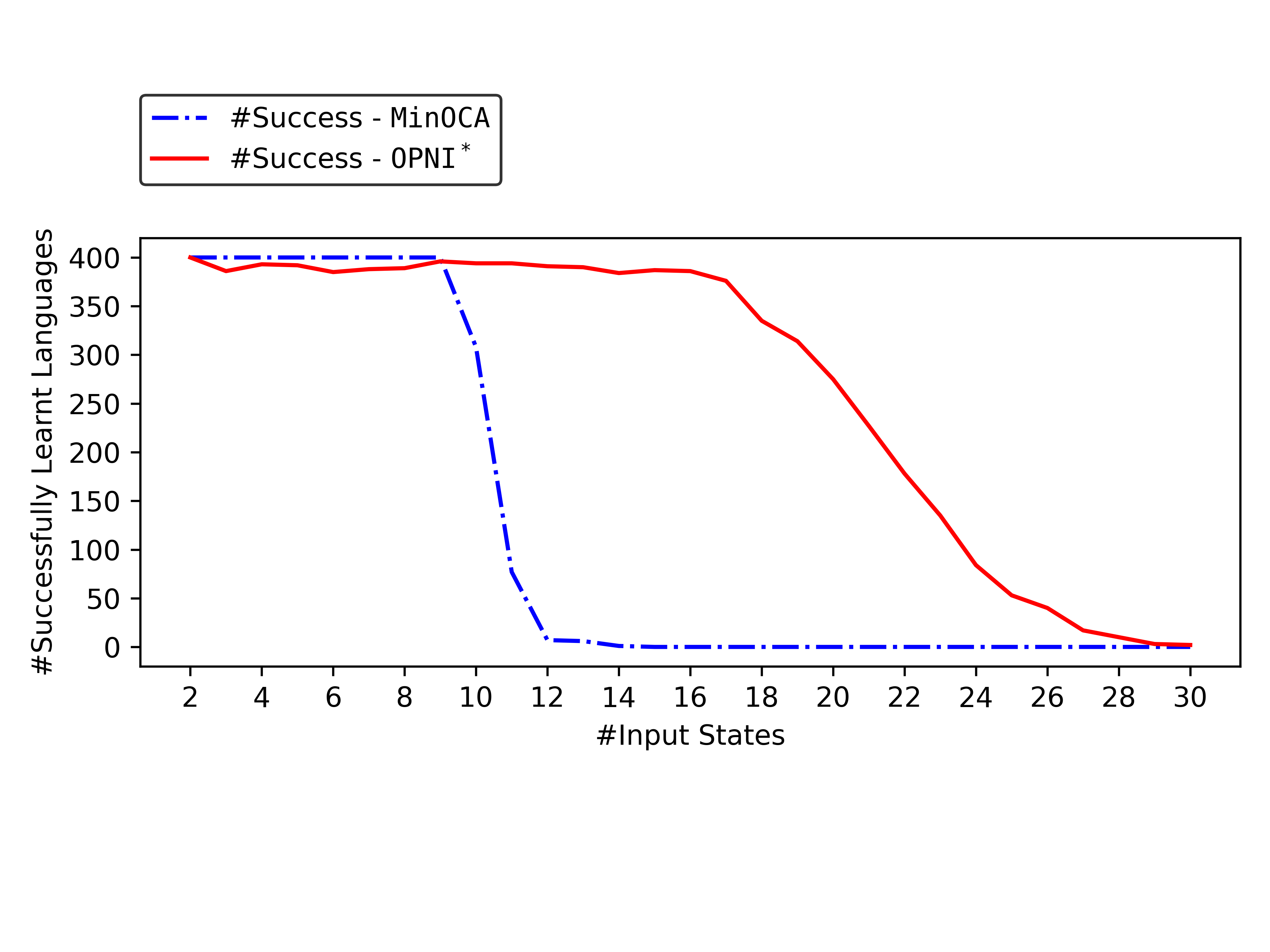}
        \caption{Number of successfully learnt languages by \opniL (Out of $400$) with $5$ minutes timeout.}
        \label{successF}
    \end{subfigure}
    \hfill
    \begin{subfigure}[t]{0.45\textwidth}
        \includegraphics[width=\linewidth]{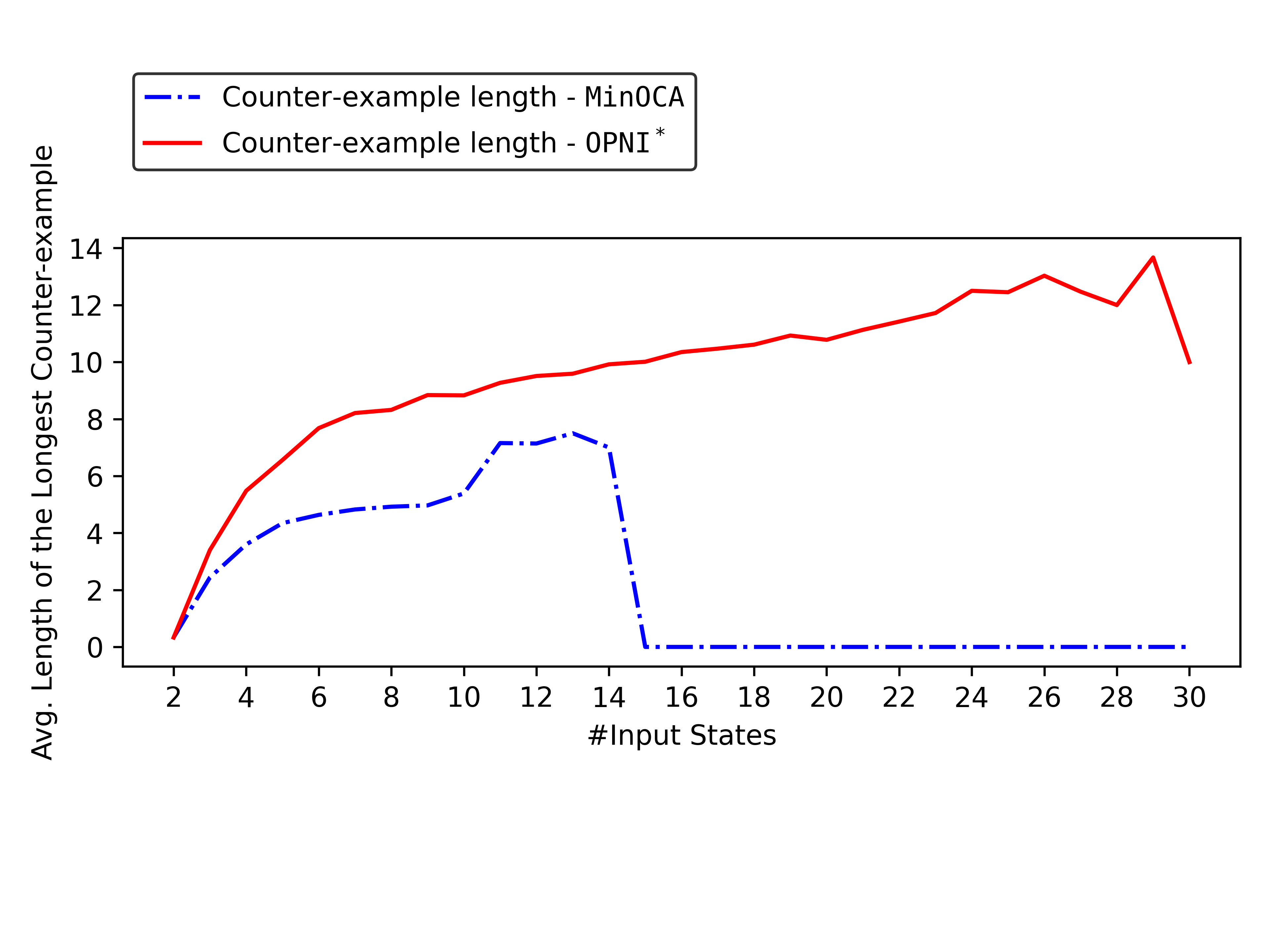}
        \caption{The average length of the longest counter-example.}
        \label{ceFig}
    \end{subfigure}

    \vspace{0.5cm}
    
    \begin{subfigure}[t]{0.45\textwidth}
        \includegraphics[width=\linewidth]{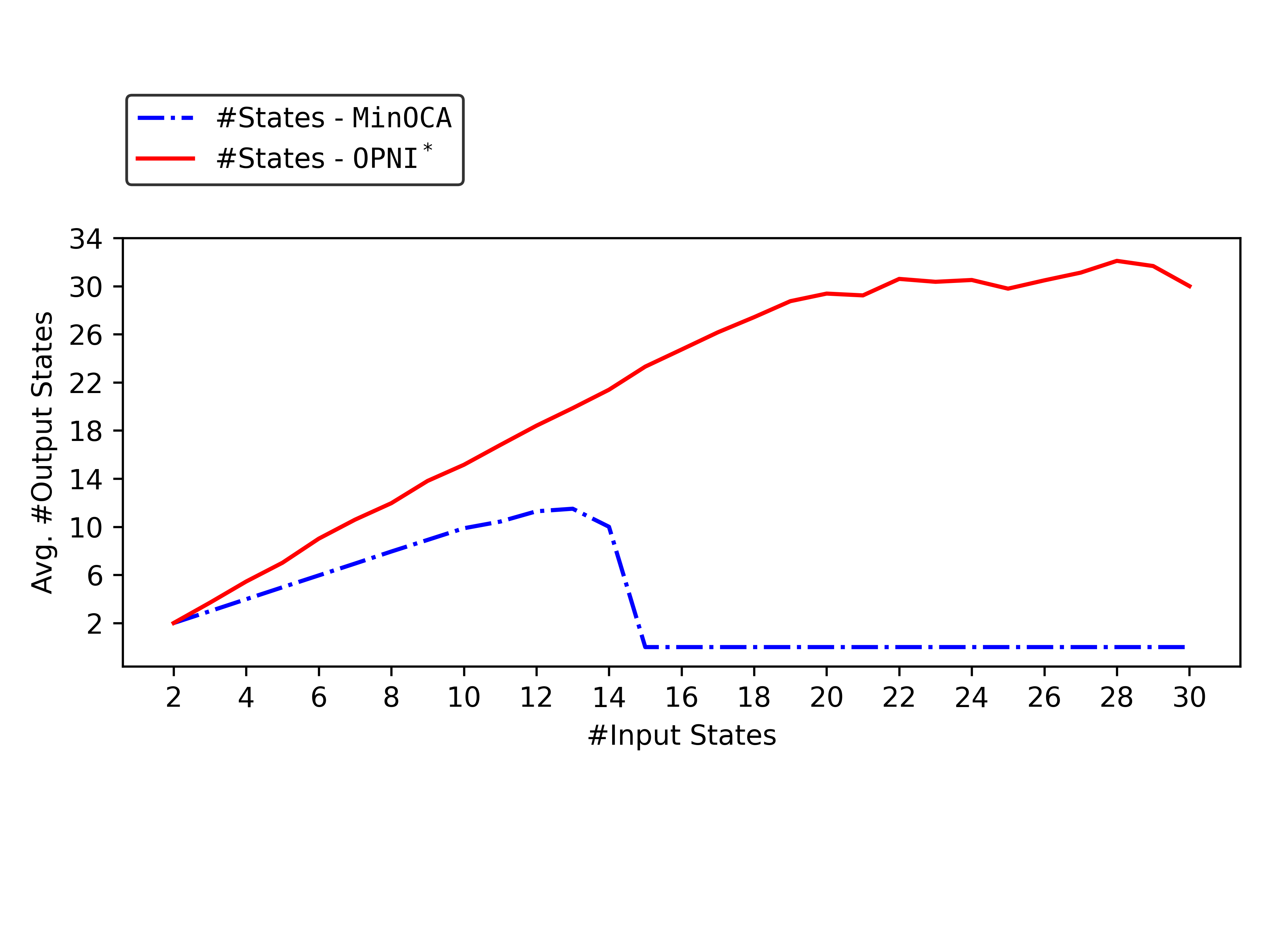}
        \caption{Average number of states in the learnt DROCA.}
        \label{StatesEquiv} 
    \end{subfigure}
    \hfill
    \begin{subfigure}[t]{0.45\textwidth}
        \includegraphics[width=\linewidth]{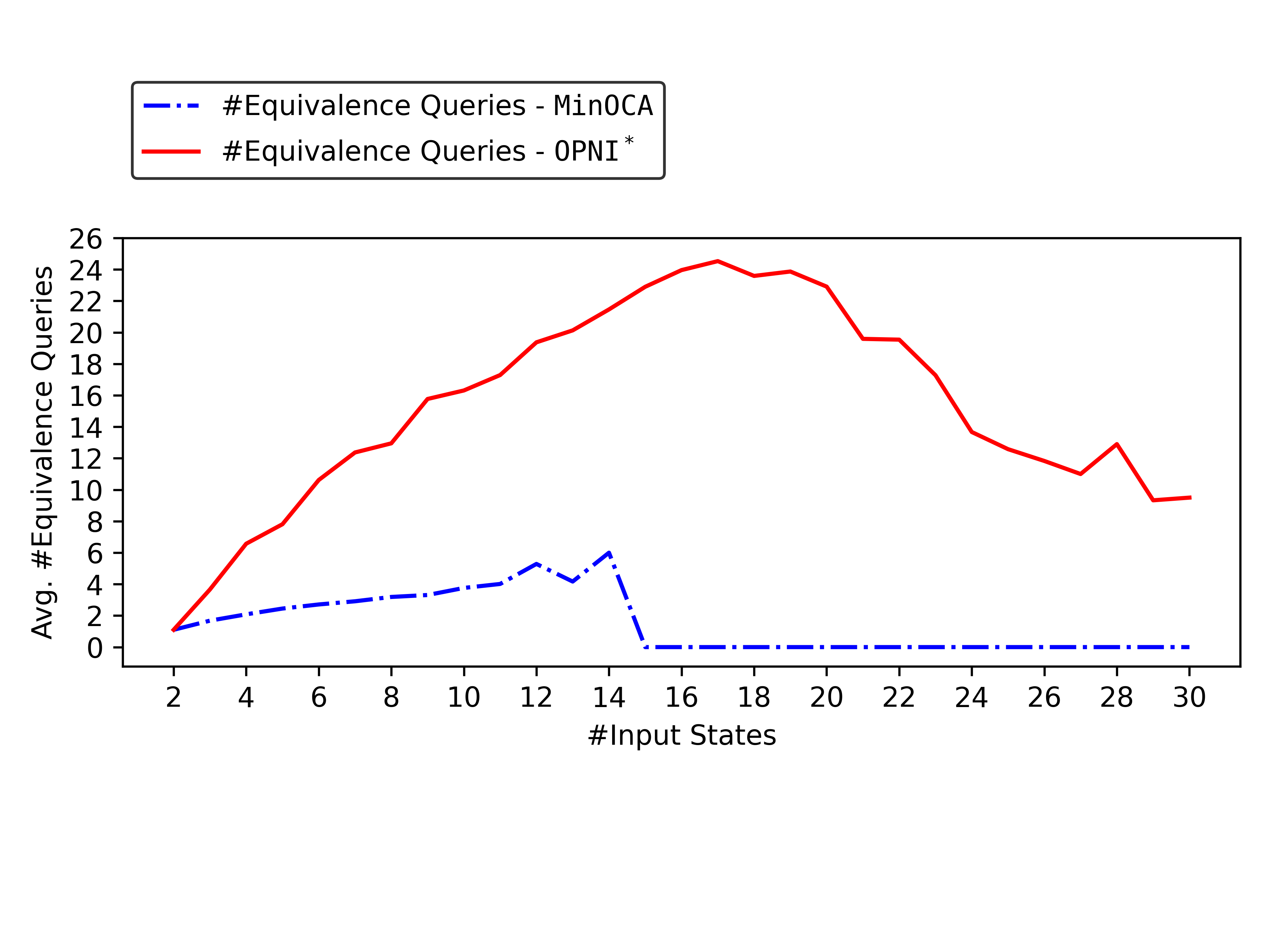}
        \caption{Average number of equivalence queries used.}
        \label{avgEquiv}
    \end{subfigure}
    \caption{Evaluation of \opniL and \minOCA for \droca on \dsOne.}
\end{figure}

\begin{figure}[h!]
    
    \begin{subfigure}[t]{0.45\textwidth}
        \includegraphics[width=\linewidth]{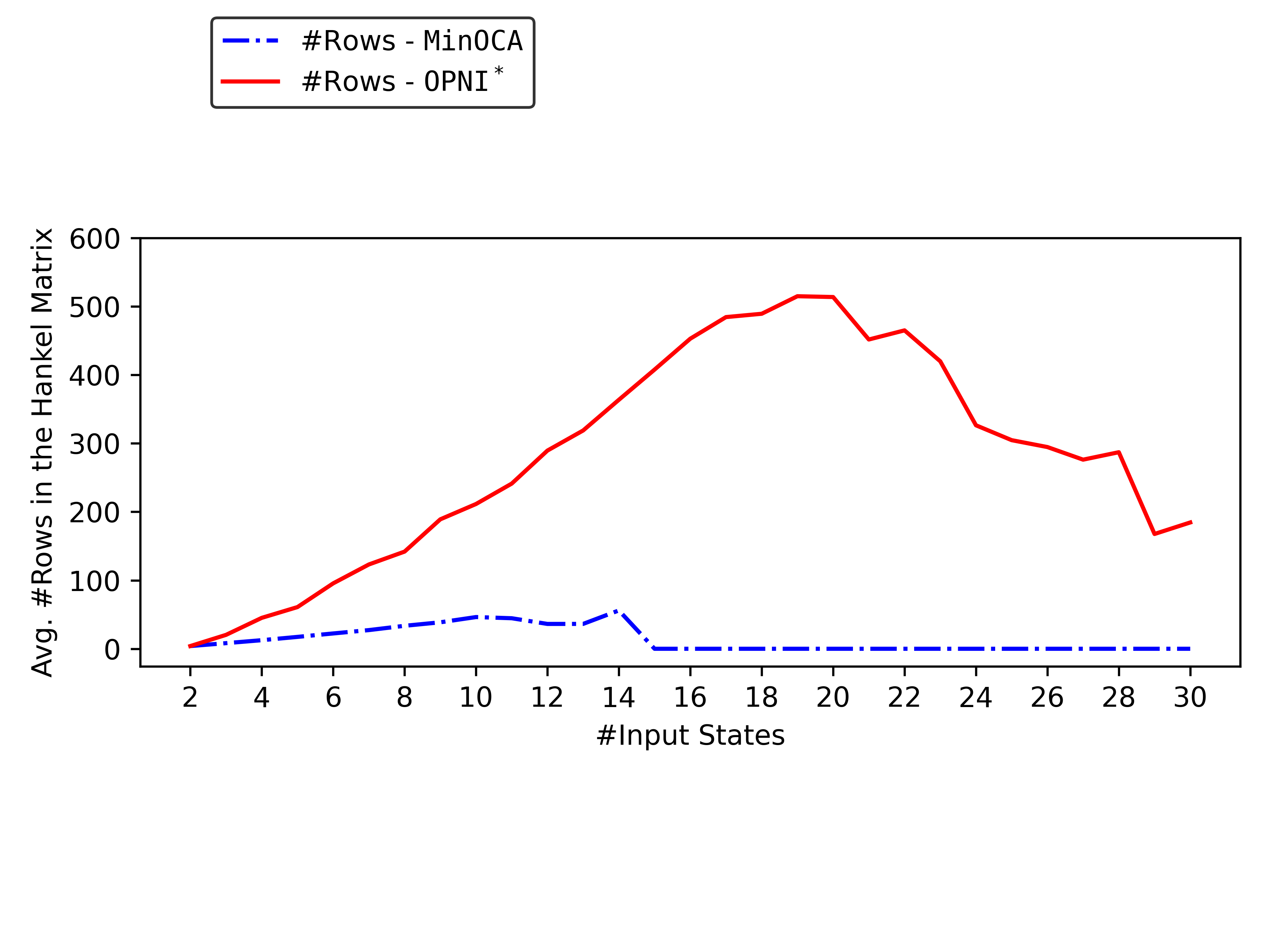}
        \caption{The average number of rows in the observation table.}
        \label{rowsD}
    \end{subfigure}
    \hfill
    \begin{subfigure}[t]{0.45\textwidth}
        \includegraphics[width=\linewidth]{./Figures/DStates.png}
        \caption{The average number of columns in the observation table.}
        \label{columnsD}
    \end{subfigure}

    \caption{Evaluation of \opniL and \minOCA for \droca on \dsOne.}
\end{figure}

\textit{1. Evaluating the performance \opniL and \minOCA on \dsOne: }
We first compare the performance of \opniL for \droca with that of \minOCA{~\cite{MathewPS25}}. 
A timeout of $5$ minutes was allotted for both $\minOCA$ and $\opniL$ for learning each \droca. 
If the procedure times out, we discard that input and process the next one. 
The number of languages successfully learnt by \opniL and \minOCA for different input sizes is depicted in \Cref{successF}. 
The proposed method outperforms \minOCA in terms of the number of successfully learnt languages within the given timeout as the number of states increases. 
The averages presented in the remaining graphs are computed using only those languages that were successfully learnt. Note that the y-values for \minOCA remain zero in all graphs for input sizes with more than $15$ states, as it fails to learn any language within the specified timeout.

\Cref{ceFig} shows the average length of the longest counter-example.
Since $\minOCA$ guarantees learning a minimal counter-synchronous \droca, whereas \opniL provides no such guarantee, the \droca learnt by \opniL typically has more states.
Consequently, the length of the counterexamples also tends to be longer for \opniL. 
 \Cref{StatesEquiv} shows the average number of states in the learnt \droca. 
 \minOCA learns a minimal counter-synchronous \droca equivalent to the input. However, the \droca learnt by \opni is equivalent and counter-synchronous with respect to the input, but not minimal. 
 \Cref{avgEquiv} shows the average number of equivalence queries used for successfully learning the input \droca. 
We also give in \Cref{rowsD} and \Cref{columnsD} the average number of rows and columns in the final observation table. 
 
Note that $\opniL$ is a semi-algorithm; its termination is not guaranteed. 
However, our experimental results on randomly generated benchmarks show that it works well in practice.

\begin{figure}[t]
    \centering
        \begin{subfigure}[t]{0.45\textwidth}
        \includegraphics[width=\linewidth]{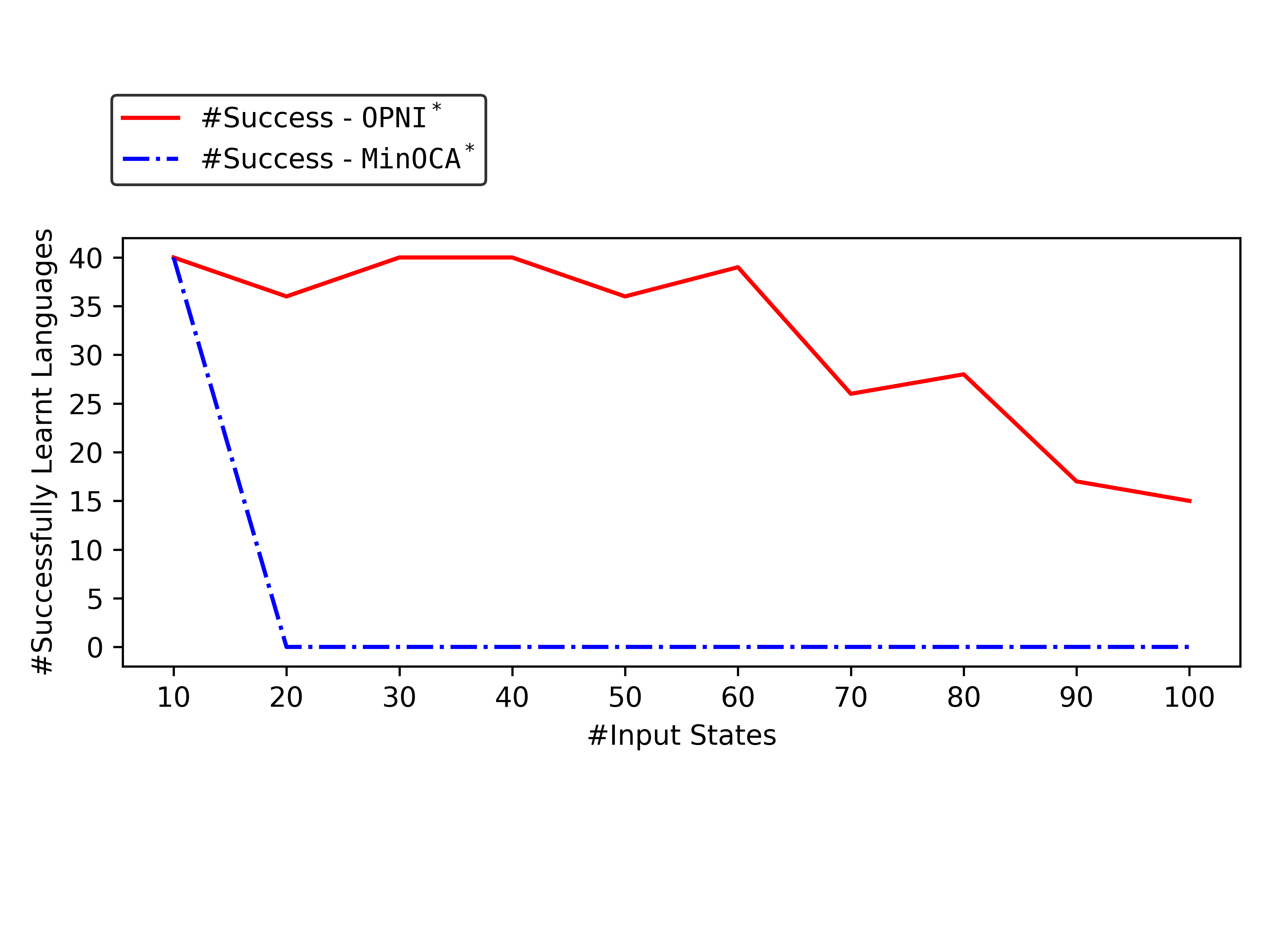}
        \caption{Number of successfully learnt languages (Out of $40$) with a timeout of $20$ minutes.}
        \label{successV}
    \end{subfigure}
    \hfill
    \begin{subfigure}[t]{0.45\textwidth}
        \includegraphics[width=\linewidth]{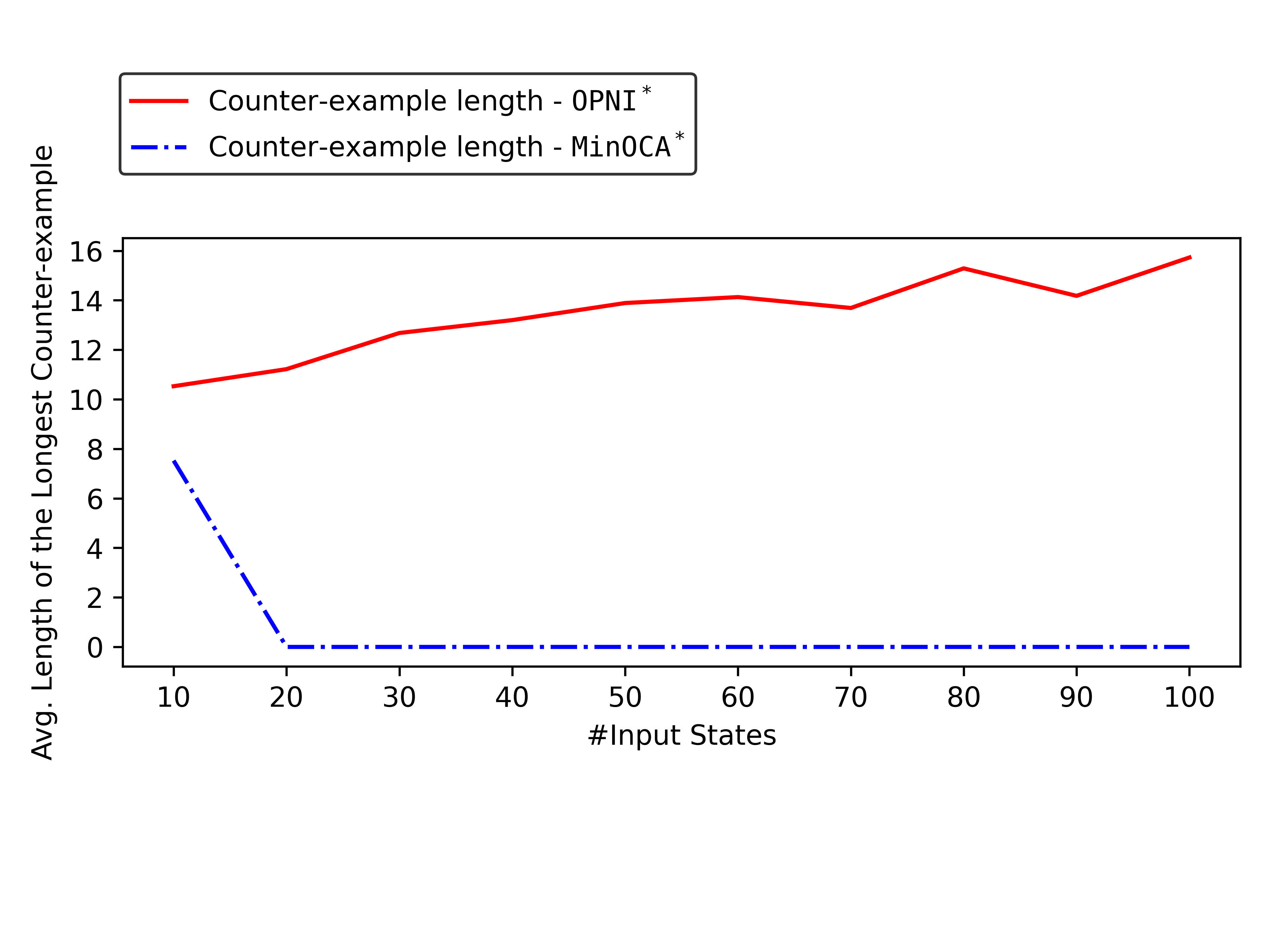}
        \caption{The average length of the longest counter-example.}
        \label{ceFigV}
    \end{subfigure}

    \vspace{0.5cm}
    
    \begin{subfigure}[t]{0.45\textwidth}
        \includegraphics[width=\linewidth]{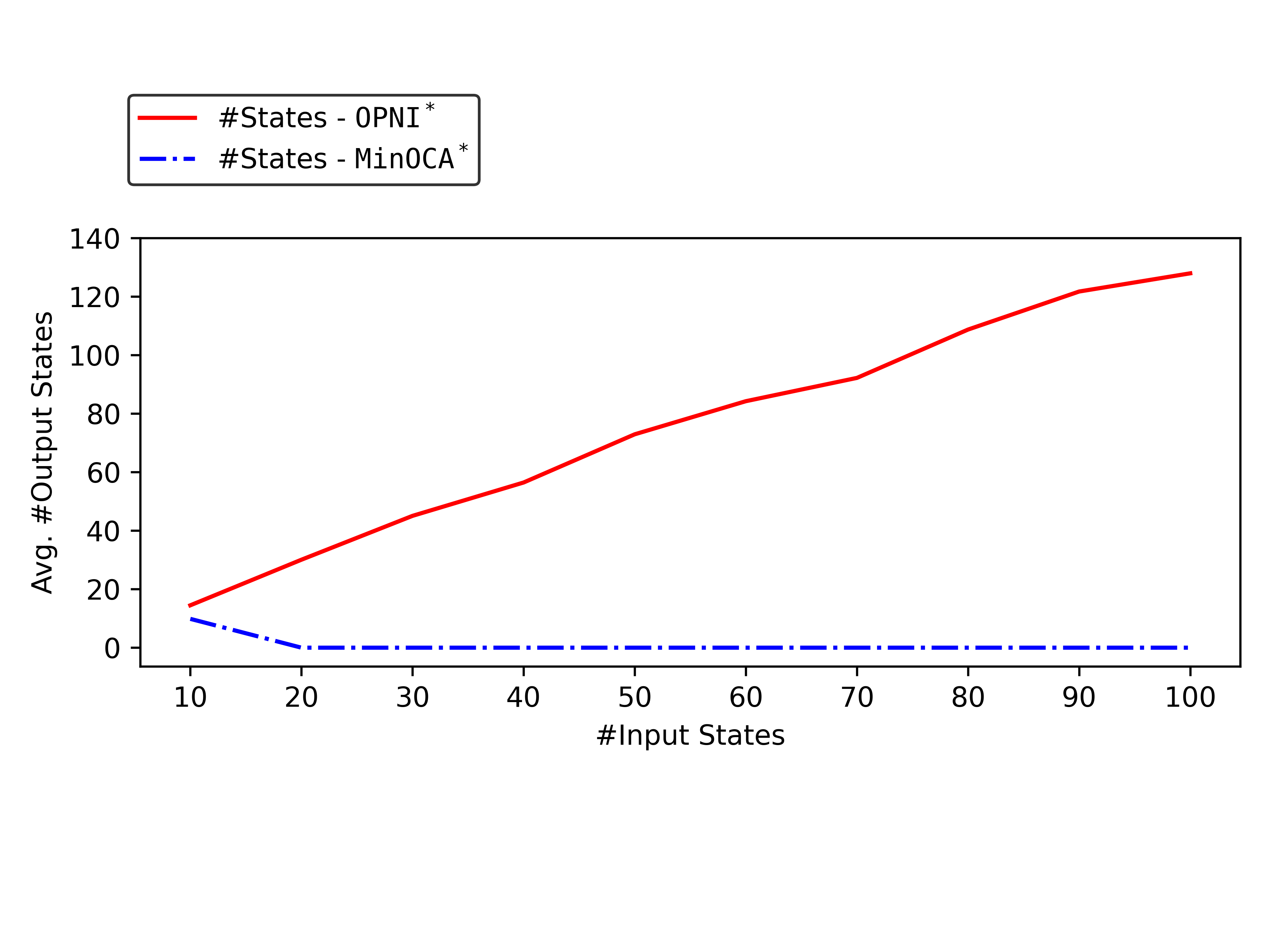}
        \caption{Average number of states in the learnt VOCA.}
        \label{StatesEqV}
    \end{subfigure}
    \hfill
    \begin{subfigure}[t]{0.45\textwidth}
        \includegraphics[width=\linewidth]{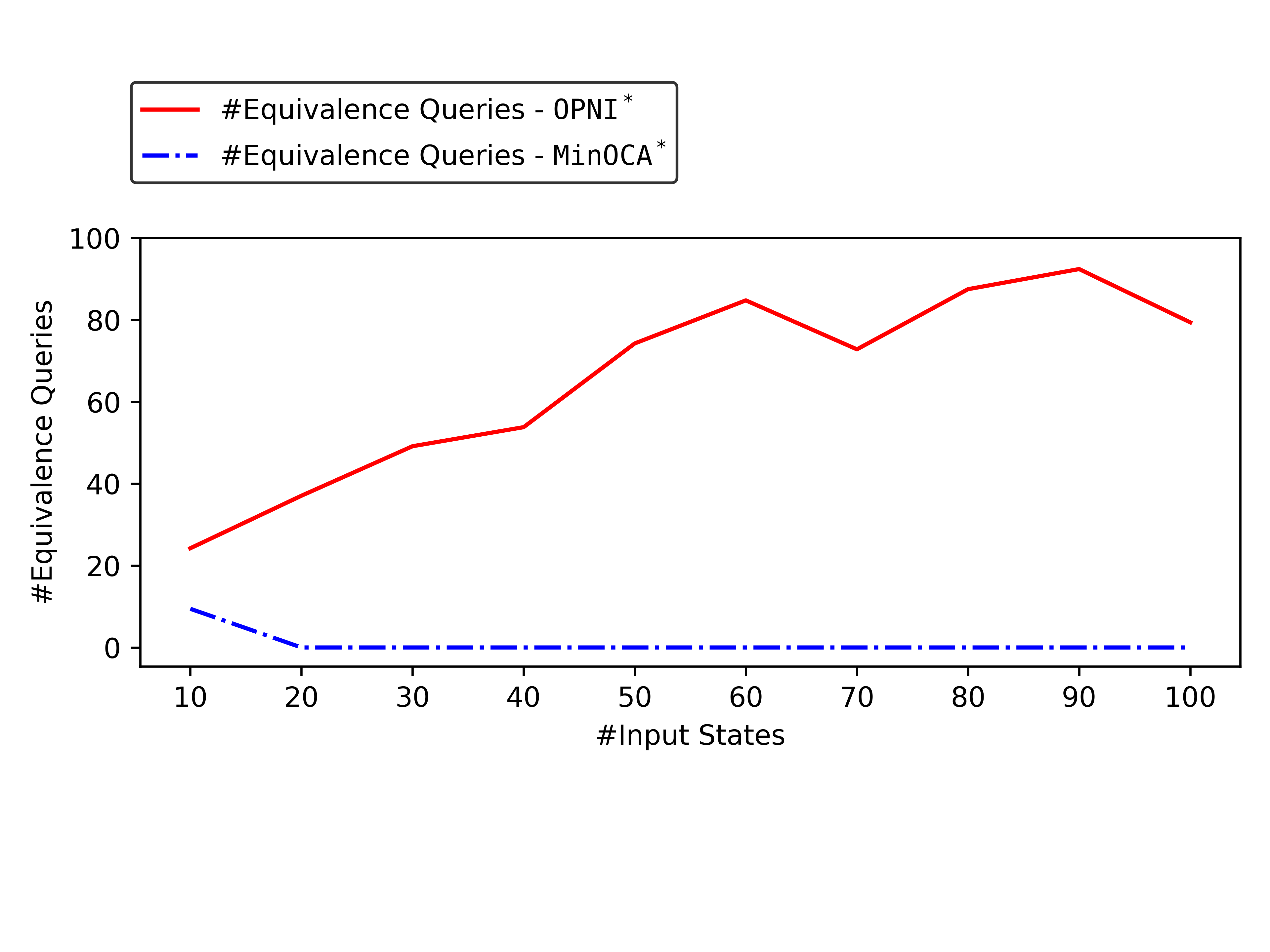}
        \caption{Number of equivalence queries used.}
        \label{avgEqV}
    \end{subfigure}

 \caption{Evaluation of \opniL and \minOCA for \voca on \dsTwo.}
    \label{ExpVoca}
\end{figure}

\begin{figure}[h!]
    \begin{subfigure}[t]{0.45\textwidth}
        \includegraphics[width=\linewidth]{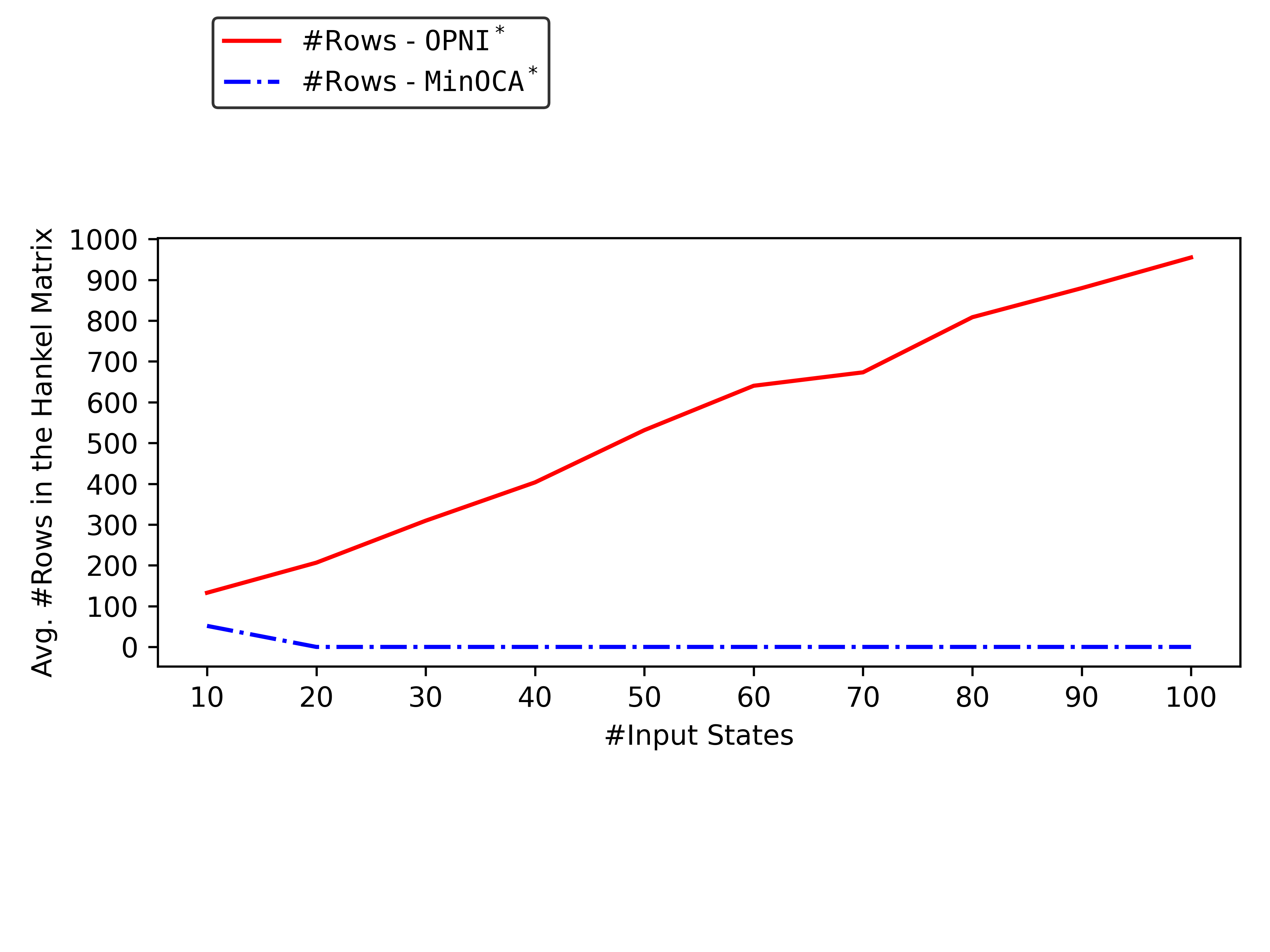}
        \caption{The average number of rows in the observation table.}
        \label{rowsV}
    \end{subfigure}
    \hfill
    \begin{subfigure}[t]{0.45\textwidth}
        \includegraphics[width=\linewidth]{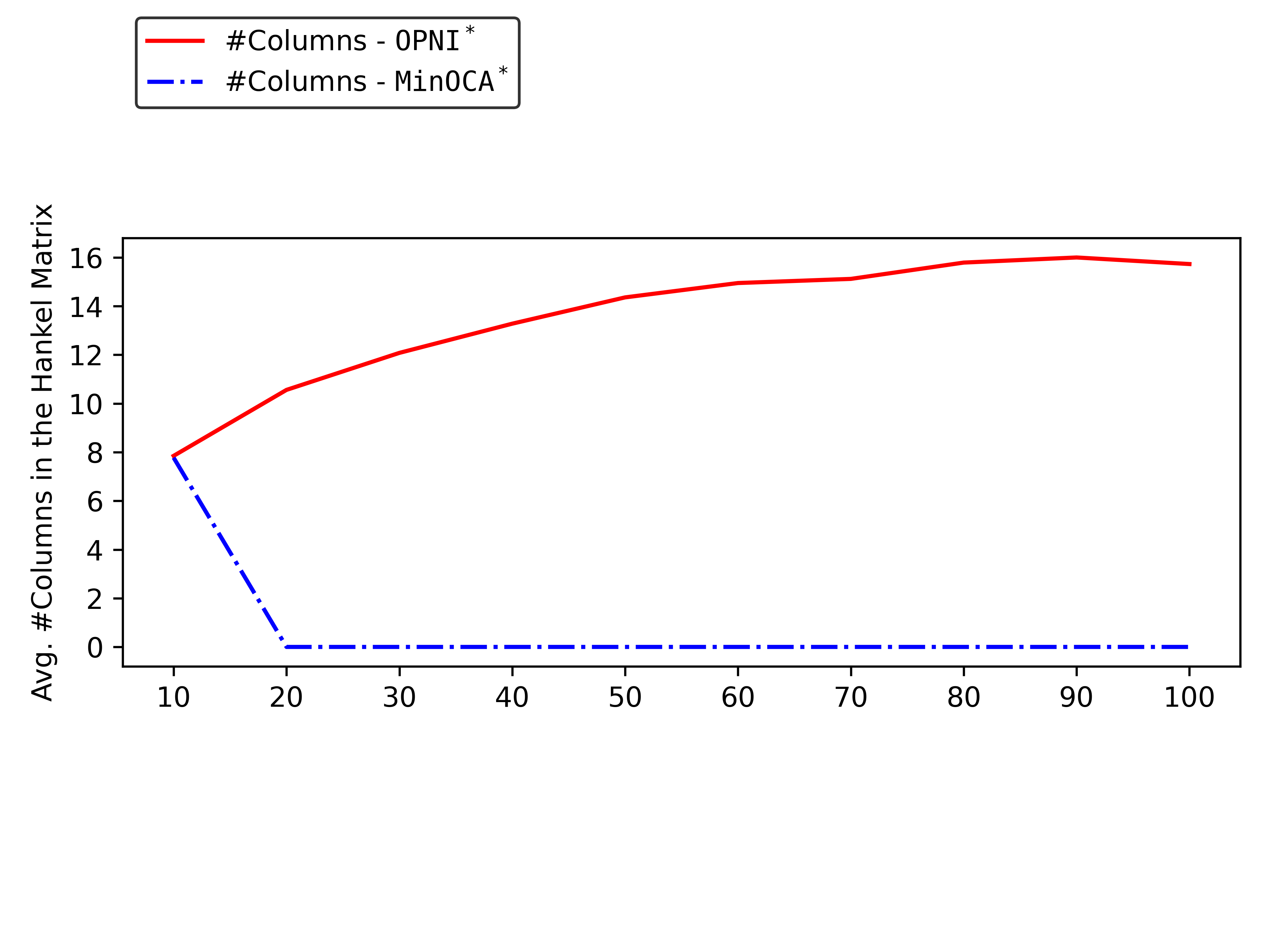}
        \caption{The average number of columns in the observation table.}
        \label{columnsV}
    \end{subfigure}

    \caption{Evaluation of \opniL and \minOCA for \voca on \dsTwo.}
\end{figure}

\textit{2. Evaluating the performance \opniL and \minOCA on \dsTwo: } 
We now compare the performance of \opniL for \voca with that of \minOCA{~\cite{MathewPS25}}. 
We note that \opniL for \voca is faster compared to \opniL for \droca. 
This is mainly due to the following two reasons: 
\begin{enumerate}
    \item faster algorithm for checking equivalence of \voca, and
    \item the input alphabet itself determines the counter actions, and therefore, step \ref{stepTwo} and step \ref{stepThree} of \Cref{alg:opni} can be skipped while passive learning \voca. 
\end{enumerate}

A similar modification of \minOCA\ {for \voca} was also implemented. 
 A timeout of $20$ minutes was allotted for both procedures for learning each \voca. If the procedure times out, we discard that input and process the next one. The number of \voca languages successfully learnt by \opniL for different input sizes is shown in \Cref{successV}. We observe that \opniL\ is able to learn close to $50\%$ of \voca\ of $100$ states. 
However, it was not able to learn any \voca with more than $20$ states. The averages presented in the remaining graphs are computed using only those languages that were successfully learnt.
Note that the y-values for \minOCA remain zero in all graphs for input sizes with more than $20$ states, as it fails to learn any language within the specified timeout.

The results of our experiments are shown in \Cref{ExpVoca}. 
\Cref{ceFigV} shows the average length of the longest counter-example. 
 \Cref{StatesEqV} shows the average number of states in the learnt \voca. 
 \Cref{avgEqV} shows the average number of equivalence queries used for successfully learning the input \voca.
 In this case also, the \voca learnt by \minOCA is minimal, whereas the one learnt by \opniL is not. 
\Cref{rowsV} and \Cref{columnsV} show the average number of rows and columns in the final observation table.

%% file: conclusion.tex
\section{Conclusion}
\label{sec:conclusion}
This work focuses on passive and active learning of deterministic real-time one-counter automata (\droca). Inspired by the classical \RPNI algorithm, we developed a passive learning algorithm, \opni, tailored to \droca. Building on this, we showed how active learning can be guided by a sequence of passive learning tasks, and proposed an active learning method, \opniL, which uses \opni\ as a subroutine. In contrast, the state-of-the-art \minOCA algorithm employs a SAT solver for this step. 

Our experiments demonstrate that \opniL scales significantly better than \minOCA in practice. We also used \opniL to learn visibly one-counter automata (\voca), observing that it can successfully learn automata with up to 100 states. Despite its practical advantages, a limitation of \opniL is that it does not guarantee termination on all inputs, since \opni\ may return increasingly large automata in the absence of minimality guarantees.

A fundamental bottleneck—common to our approach and to prior works such as~\cite{bps, MathewPS25}—is the reliance on counter-value queries. This renders the learning process a grey-box framework, in contrast to Angluin’s black-box $L^*$ algorithm for finite automata. Recent work~\cite{LearningInP} has shown that active learning of \droca\ is possible in polynomial time even without counter-value queries. However, the proposed algorithm is not practical due to the high-degree polynomials involved. An interesting direction for future work is to combine insights from the methods: \opniL, \minOCA~\cite{MathewPS25}, and $\textsf{OL}^*$~\cite{LearningInP} to develop learning algorithms that are both theoretically efficient and practically usable.

Finally, extending these ideas to richer models such as pushdown automata—or their subclasses, such as visibly pushdown automata—offers a direction for future research.

%% file: fsttcs_main.bbl
\begin{thebibliography}{10}

\bibitem{AM04}
Rajeev Alur and P.~Madhusudan.
\newblock Visibly pushdown languages.
\newblock In L{\'{a}}szl{\'{o}} Babai, editor, {\em Proceedings of the 36th
  Annual {ACM} Symposium on Theory of Computing, Chicago, IL, USA, June 13-16,
  2004}, pages 202--211. {ACM}, 2004.

\bibitem{AM09}
Rajeev Alur and P.~Madhusudan.
\newblock Adding nesting structure to words.
\newblock {\em J. {ACM}}, 56(3):16:1--16:43, 2009.

\bibitem{Angluin87}
Dana Angluin.
\newblock Learning regular sets from queries and counterexamples.
\newblock {\em Inf. Comput.}, 75(2):87--106, 1987.

\bibitem{droca}
Stanislav B{\"{o}}hm and Stefan G{\"{o}}ller.
\newblock Language equivalence of deterministic real-time one-counter automata
  is nl-complete.
\newblock In {\em {MFCS}}, volume 6907 of {\em Lecture Notes in Computer
  Science}, pages 194--205. Springer, 2011.

\bibitem{bps}
V{\'{e}}ronique Bruy{\`{e}}re, Guillermo~A. P{\'{e}}rez, and Ga{\"{e}}tan
  Staquet.
\newblock Learning realtime one-counter automata.
\newblock In {\em {TACAS} {(1)}}, volume 13243 of {\em Lecture Notes in
  Computer Science}, pages 244--262. Springer, 2022.

\bibitem{FahmyR95}
Amr~F. Fahmy and Robert~S. Roos.
\newblock Efficient learning of real time one-counter automata.
\newblock In {\em {ALT}}, volume 997 of {\em Lecture Notes in Computer
  Science}, pages 25--40. Springer, 1995.

\bibitem{minOCA}
Prince Mathew.
\newblock princemathew07/minoca: Minoca v1.0.0, January 2025.
\newblock \href {https://doi.org/10.5281/zenodo.14604419}
  {\path{doi:10.5281/zenodo.14604419}}.

\bibitem{LearningInP}
Prince Mathew, Vincent Penelle, and A.~V. Sreejith.
\newblock Learning deterministic one-counter automata in polynomial time.
\newblock In {\em the proceedings of the 40th Annual ACM/IEEE Symposium on
  Logic in Computer Science (LICS 2025) (to appear)}, 2025.
\newblock URL: \url{https://arxiv.org/abs/2503.04525}, \href
  {https://arxiv.org/abs/2503.04525} {\path{arXiv:2503.04525}}.

\bibitem{MathewPS25}
Prince Mathew, Vincent Penelle, and A.~V. Sreejith.
\newblock Learning real-time one-counter automata using polynomially many
  queries.
\newblock In {\em {TACAS} {(1)}}, volume 15696 of {\em Lecture Notes in
  Computer Science}, pages 276--294. Springer, 2025.

\bibitem{christof}
Daniel Neider and Christof L\"{o}ding.
\newblock Learning visibly one-counter automata in polynomial time.
\newblock {\em Technical Report, RWTH Aachen}, AIB-2010-02, 2010.

\bibitem{oncina1992}
Jos{\'e} Oncina and Pedro Garcia.
\newblock Inferring regular languages in polynomial updated time.
\newblock In {\em Pattern recognition and image analysis: selected papers from
  the IVth Spanish Symposium}, pages 49--61. World Scientific, 1992.

\bibitem{VP75}
Leslie~G. Valiant and Mike Paterson.
\newblock Deterministic one-counter automata.
\newblock {\em J. Comput. Syst. Sci.}, 10(3):340--350, 1975.

\end{thebibliography}
